\documentclass[12pt,onecolumn,romanappendices]{IEEEtran}

\usepackage{cite}
\usepackage{amsmath,amssymb,amsfonts,amscd,amsthm,amsbsy}
\usepackage{algorithmic}
\usepackage{graphicx}
\usepackage{textcomp}
\usepackage{wrapfig}
\usepackage{dsfont}
\usepackage[english]{babel}
\usepackage[ruled,vlined]{algorithm2e}
\usepackage{threeparttable}
\usepackage{textcomp}
\usepackage{multirow}
\usepackage{mathrsfs}
\usepackage{bbm}
\usepackage{mathtools}
\usepackage{array}
\usepackage{multirow}
\usepackage{capt-of}
\usepackage{caption}
\usepackage{subcaption}
\newtheorem{theorem}{Theorem}
\newtheorem{lemma}{Lemma}

\newtheorem{corollary}{Corollary}
\theoremstyle{definition}
\newtheorem{remark}{Remark}
\newtheorem{mydef}{Definition}
\theoremstyle{definition}

\theoremstyle{definition}
\newtheorem{example}{Example}

\theoremstyle{definition}

\theoremstyle{definition}
\newtheorem{opt}{Optimization}
\makeatletter
\renewcommand*\env@matrix[1][c]{\hskip -\arraycolsep
  \let\@ifnextchar\new@ifnextchar
  \array{*\c@MaxMatrixCols #1}}
\makeatother

\usepackage{etoolbox}
\AtEndEnvironment{example}{\null\hfill\IEEEQED}%

\newcommand{\Fb}{\mathds{F}}
\newcommand{\xcal}{\mathcal{X}}
\newcommand{\ycal}{\mathcal{Y}}
\newcommand{\acal}{\mathcal{A}}
\newcommand{\ical}{\mathcal{I}}
\newcommand{\bcal}{\mathcal{B}}
\newcommand{\ucal}{\mathcal{U}}
\newcommand{\ecal}{\mathcal{E}}
\newcommand{\pcal}{\mathcal{P}}

\newcommand{\cs}{\mathsf{C}}
\newcommand{\Lb}{{\bf{L}}}
\newcommand{\xb}{{\bf{x}}}
\newcommand{\zb}{{\bf{z}}}
\newcommand{\yb}{{\bf{y}}}
\newcommand{\Hb}{{\bf{H}}}

\newcommand{\cb}{{\bf{c}}}

\newcommand{\Comb}[2]{{}^{#1}C_{#2}}%
\newcommand{\Ab}{{\textbf{A}}}
\newcommand{\ebh}{\hat{\pmb{\epsilon}}}
\newcommand{\bb}{{\pmb{\beta}}}
\newcommand{\ei}{{\pmb{\epsilon}}_{i}}
\newcommand{\eb}{{\pmb{\epsilon}}}

\usepackage{psfrag}

\usepackage{paralist}

\usepackage{xfrac}

\title{Linear Codes for Broadcasting with Noisy Side Information}

\author{\IEEEauthorblockN{Suman Ghosh and Lakshmi Natarajan}

\thanks{The authors are with Department of Electrical Engineering, Indian Institute of Technology Hyderabad, Sangareddy 502\,285, India (email: \{ee16resch11006,\,lakshminatarajan\}@iith.ac.in).}}%

\begin{document}

\maketitle

\begin{abstract}
We consider network coding for a noiseless broadcast channel where each receiver demands a subset of messages available at the transmitter and is equipped with \emph{noisy side information} in the form an erroneous version of the message symbols it demands. We view the message symbols as elements from a finite field and assume that the number of symbol errors in the noisy side information is upper bounded by a known constant. This communication problem, which we refer to as \emph{broadcasting with noisy side information (BNSI)}, has applications in the re-transmission phase of downlink networks. 
We derive a necessary and sufficient condition for a linear coding scheme to satisfy the demands of all the receivers in a given BNSI network, and show that syndrome decoding can be used at the receivers to decode the demanded messages from the received codeword and the available noisy side information.
We represent BNSI problems as bipartite graphs, and using this representation, classify the family of problems where linear coding provides bandwidth savings compared to uncoded transmission. 
We provide a simple algorithm to determine if a given BNSI network belongs to this family of problems, i.e., to identify if linear coding provides an advantage over uncoded transmission for the given BNSI problem.
We provide lower bounds and upper bounds on the optimal codelength and constructions of linear coding schemes based on linear error correcting codes. For any given BNSI problem, we construct an equivalent index coding problem. A linear code is a valid scheme for a BNSI problem if and only if it is valid for the constructed index coding problem.
\end{abstract}


\section{Introduction} \label{Intro}
We consider the problem of broadcasting $n$ message symbols $x_1,\dots,x_n$ from a finite field $\Fb_q$ to a set of $m$ users $u_1,\dots,u_m$ through a noiseless broadcast channel. The $i^{\text{th}}$ receiver $u_i$ requests the message vector $\xb_{\xcal_i} = (x_j, \, j \in \xcal_i)$ where \mbox{$\xcal_i \subseteq \{1,\dots,n\}$} denotes the demands of $u_i$. We further assume that each receiver knows a noisy/erroneous version of \emph{its own demanded message} as side information. In particular, we assume that the side information at $u_i$ is a $\Fb_q$-vector $\xb_{\xcal_i}^{e}$ such that the demanded message vector $\xb_{\xcal_i}$ differs from the side information $\xb_{\xcal_i}^{e}$ in at the most $\delta_s$ coordinates, where the integer $\delta_s$ determines the quality of side information.
We assume that the transmitter does not know the exact realizations of the side information vectors available at the receivers. The objective of code design is to broadcast a codeword of as small a length as possible such that every receiver can retrieve its demanded message vector using the transmitted codeword and the available noisy side information. We refer to this communication problem as \emph{broadcasting with noisy side information} (BNSI). 

Wireless broadcasting in downlink communication channels has gained considerable attention and has several important applications, such as cellular and satellite communication, digital video broadcasting, and wireless sensor networks.
The BNSI problem considered in this paper models the re-transmission phase of downlink communication channels at the network layer.
Suppose during the initial broadcast phase each receiver of a downlink network decodes its demanded message packet erroneously (such as when the wireless channel experiences outage). 
Instead of discarding this decoded message packet, the erroneous symbols from this packet can be used as noisy side information for the re-transmission phase. 
If the number of symbol errors $\delta_s$ in the erroneously decoded packets is not large, we might be able to reduce the number of channel uses required for the re-transmission phase by intelligently coding the message symbols at the network layer. 

Consider the example scenario shown in Fig.~\ref{fig:subim1}. The transmitter is required to broadcast 4 message symbols $x_1,x_2,x_3,x_4$ to 3 users. Each user requires a subset of the message symbols, for example, User 1, User 2 and User 3 demand $(x_1,x_2,x_3)$, $(x_2,x_3,x_4)$ and $(x_1,x_3,x_4)$, respectively. Suppose during the initial transmission the broadcast channel is in outage, as experienced during temporary weather conditions in satellite-to-terrestrial communications. As a result, at each user, one of the message symbols in the decoded packet is in error. Based on an error detection mechanism (such as cyclic redundancy check codes) all the users request for a re-transmission. We assume that the users are not aware of the position of the symbol errors. 
\begin{figure}[t]
\begin{subfigure}{0.5\textwidth}
\centering
\includegraphics[width=3.25in]{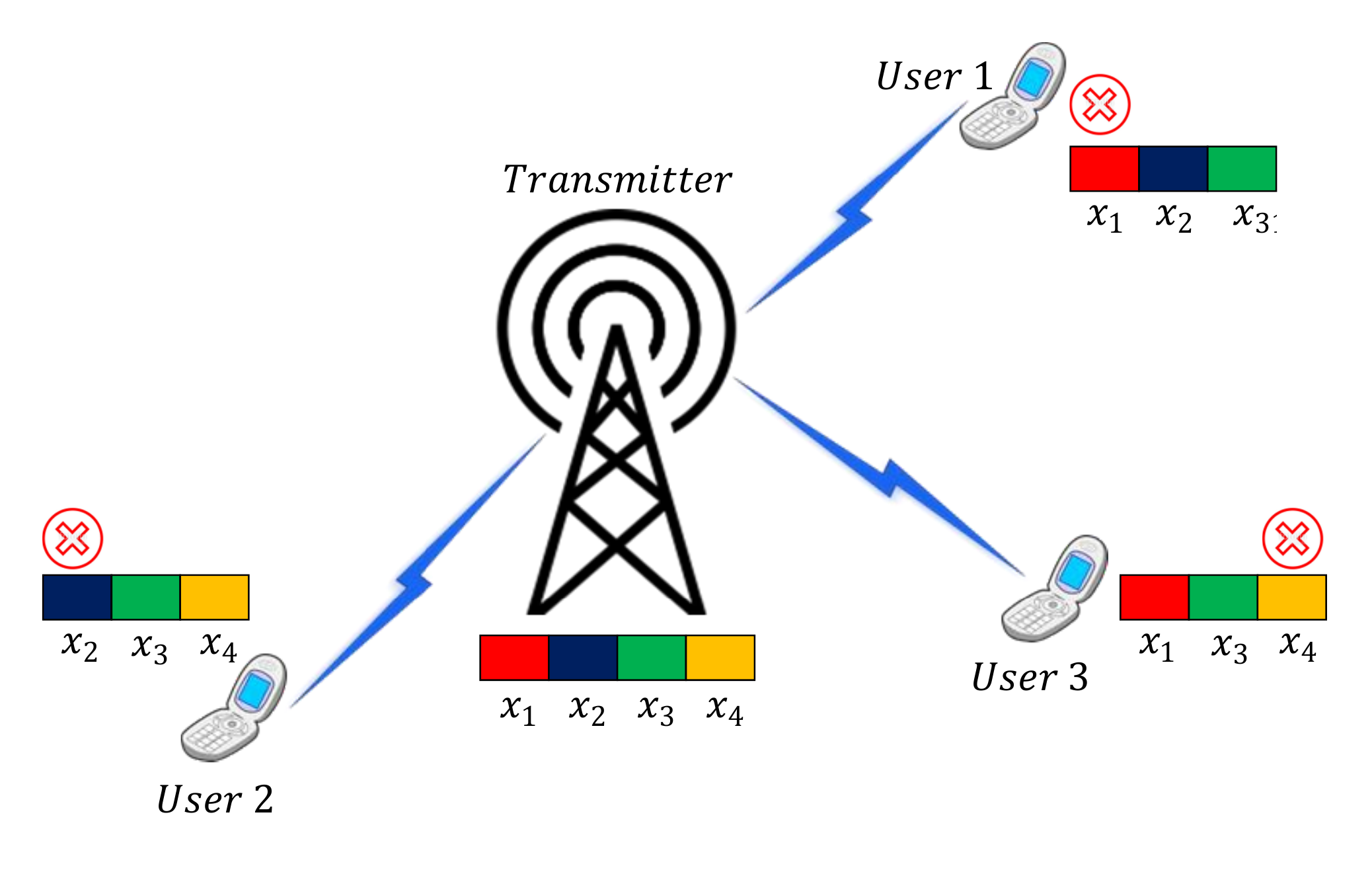} 
\caption{Transmission phase}
\label{fig:subim1}
\end{subfigure}
\begin{subfigure}{0.5\textwidth}
\centering
\includegraphics[width=3.25in]{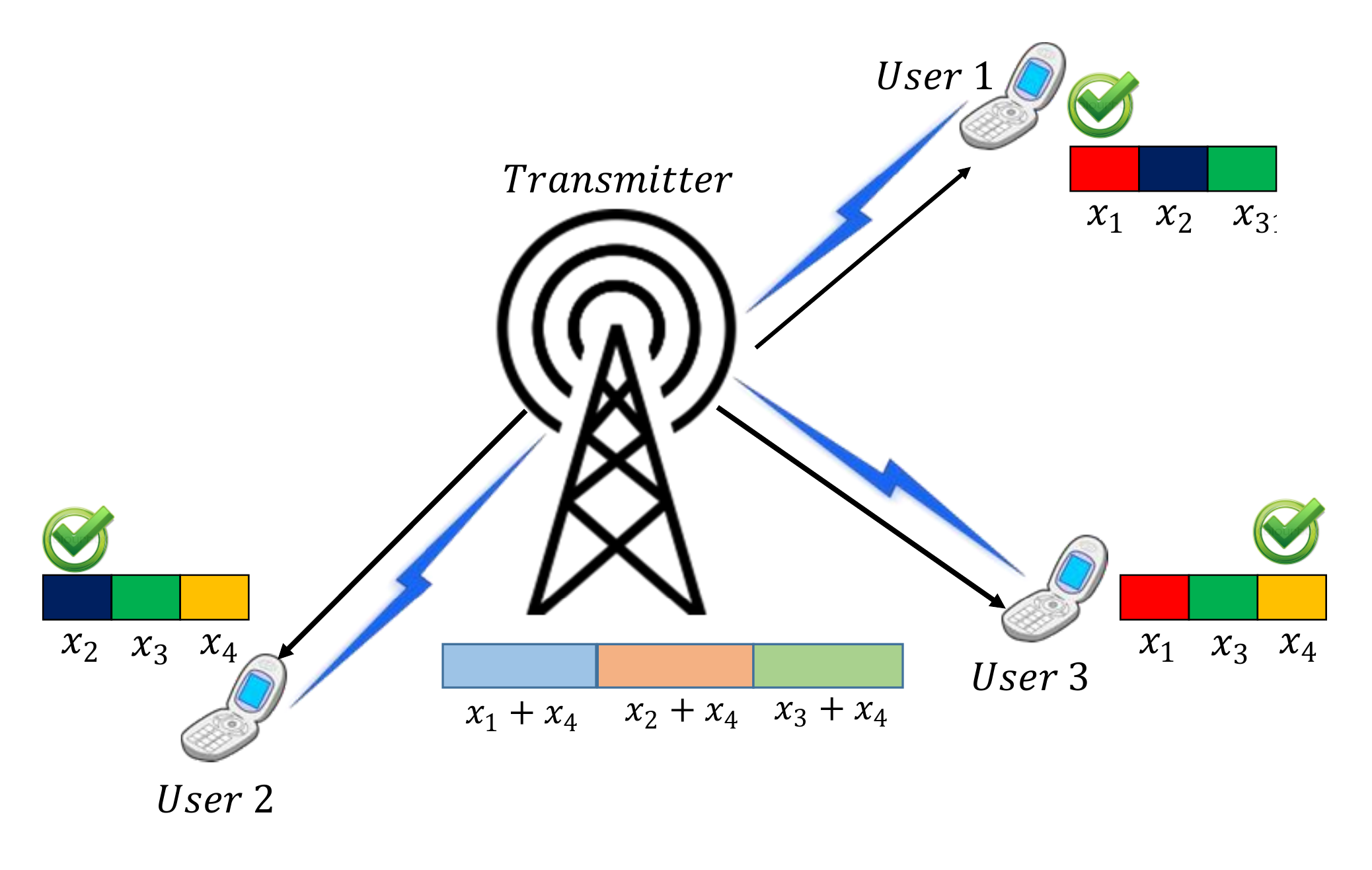}
\caption{Retransmission phase}
\label{fig:subim2}
\end{subfigure}
\caption{An example of broadcast channel with noisy side information.}
\label{fig:image0}
\end{figure}
The transmitter attempts a retransmission when the channel conditions improve.
Instead of retransmitting each message packet individually, which will require $4$ symbols to be transmitted, the transmitter will broadcast the coded sequence $(x_1+x_4,x_2+x_4,x_3+x_4)$ consisting of $3$ symbols, as shown in Fig.~\ref{fig:subim2}. 
Upon receiving this coded sequence it can be shown that using appropriate decoding algorithm (Examples~\ref{exmp2} and~\ref{exmp3} in Sections~\ref{sec3} and~\ref{sec4}, respectively) each user can correctly retrieve its own demanded message symbols using the erroneous version that it already has. 
By using a carefully designed code the transmitter is be able to reduce the number of channel uses in the retransmission phase. 

\subsection{Related Work}
\emph{Index coding}~\cite{YBJK_IEEE_IT_11} is a related code design problem that is concerned with the transmission of a set of information symbols to finitely many receivers in a noiseless broadcast channel where each receiver demands a subset of information symbols from the transmitter and already knows a \emph{different} subset of symbols as side information. The demand subset and the side information subset at each receiver in index coding are disjoint and the side information is assumed to be noiseless. 
Several results on index coding are available based on algebraic and graph theoretic formulations~\cite{linear_prog_10,MCJ_IT_14,VaR_GC_16,
MBBS_2016,MaK_ISIT_17,tehrani2012bipartite,SDL_ISIT_13,
LOCKF_2016,MAZ_IEEE_IT_13,agar_maz_2016}. 
The problem of index coding under noisy broadcast channel conditions has also been studied. Dau et al.~\cite{Dau_IEEE_J_IT_13} analyzed linear index codes for error-prone broadcast channels. Several works, for example~\cite{karat_rajan_2017,samuel_rajan_2017}, provide constructions of error correcting index codes. Kim and No~\cite{Kim_No_2017} consider errors both during broadcast channel transmission as well as in receiver side information.

Index coding achieves bandwidth savings by requiring each receiver to know a subset of messages that it does not demand from the source. This side information might be gathered by overhearing previous transmissions from the source to other users in the network.
In contrast, the coding scenario considered in this paper does not require a user to overhear and store data packets that it does not demand (which may incur additional storage and computational effort at the receivers), but achieves bandwidth savings by exploiting the erroneous symbols already available from prior failed transmissions to the same receiver.
To the best of our knowledge, no code design criteria, analysis of code length or code constructions are available for the class of broadcast channels with noisy side information considered in this paper.

\subsection{Contributions and Organization}

We view broadcasting with noisy side information as a coding theoretic problem at the network layer. We introduce the system model and provide relevant definitions in Section~\ref{sec2}. We consider linear coding schemes for the BNSI problem and provide a necessary and sufficient condition for a linear code to meet the demands of all the receivers in the broadcast channel (Theorem~\ref{thm1} and Corollary~\ref{corr1}, Section~\ref{sec3}).  Given a linear coding scheme for a BNSI problem, we show how each receiver can decode its demanded message from the transmitted codeword and its noisy side information using the syndrome decoding technique (Section~\ref{sec4}). We then provide an exact characterization of the family of BNSI problems where the number of channel uses required with linear coding is strictly less than that required by uncoded transmission (Theorem~\ref{thm2}, Section~\ref{sec5b}). We provide a simple algorithm to determine if a given BNSI network belongs to this family of problems using a representation of the problem in terms of a bipartite graph (Algorithm~2, Section~\ref{sec5c}). Next we provide lower bounds on the optimal codelength (Section \ref{lb}). A simple construction of an encoder matrix based on linear error correcting code is described (Section \ref{linecc}). Based on this construction we then provide upper bounds on the optimal codelength (Section \ref{ub}). Finally we relate the BNSI problem with  index coding problem. We show that each BNSI problem is equivalent to an index coding problem (Section \ref{ic_bnsi}). We show that any linear code is a valid coding scheme for a BNSI problem if and only if it is valid for the equivalent index coding problem (Theorem~\ref{BNSI_IC_L}, Section~\ref{ic_bnsi}).
A lower bound on optimal codelength of a BNSI problem is also derived from the equivalent index coding problem (Section \ref{lb_ic}).  

\emph{Notation}: Matrices and row vectors are denoted by bold uppercase and lowercase letters, respectively. For any positive integer $n$, the symbol $[n]$ denotes the set $\{1,\dots,n\}$. The Hamming weight of a vector $\xb$ is denoted as $wt(\xb)$. The symbol $\Fb_q$ denotes the finite field of size $q$, where $q$ is a prime power. The $n \times n$ identity matrix is denoted as ${\bf I}_n$. For any matrix ${\bf L} \in \Fb_q^{n \times N}$, $rowspan\{{\bf L}\}$ denotes the subspace of $\Fb_q^N$ spanned by the rows of ${\bf L}$, and $\Lb^T$ is the transpose of $\Lb$.
  
\section{System Model and Definitions} \label{sec2}
Suppose a transmitter intends to broadcast a vector of $n$ information symbols from a finite field $\Fb_q$ denoted as \mbox{$\xb=(x_1,x_2,\dots,x_n) \in \Fb_q^n$} to $m$ users or receivers denoted as $u_1,u_2,\dots,u_m$. The demanded information symbol vector of $i^{\text{th}}$ user $u_i$ is denoted as \mbox{$\xb_{\xcal_i}=(x_j,j \in \xcal_i) \in \Fb_q^{|\xcal_i|}$} where $\xcal_i \subseteq [n]$ is the \textit{demanded information symbol index set} of the $i^{\text{th}}$ user. The $m$-tuple \mbox{$\xcal=(\xcal_1,\xcal_2,\dots,\xcal_m)$} represents the demands of all the $m$ receivers in the broadcast channel. The erroneous version of the demanded information symbol vector available as side information at user $u_i$ is denoted as \mbox{$\xb_{\xcal_i}^{e} \in \Fb_q^{|\xcal_i|}$}. We will assume that the noisy side information $\xb_{\xcal_i}^{e}$ differs from the actual demanded message vector $\xb_{\xcal_i}$ in at the most $\delta_s$ coordinates, i.e. \mbox{$\xb_{\xcal_i}^{e}=\xb_{\xcal_i}+\ei$}, where the noise vector $\ei \in \Fb_q^{|\xcal_i|}$ and \mbox{$wt(\ei) \leq \delta_s$}. We will further assume that the transmitter and all the receivers know the value of $\delta_s$ and $\xcal$, but are unaware of the exact realization of the noise vectors $\eb_1,\dots,\eb_m$.

The coding problem considered in this paper is to generate a transmit codeword $\cb=(c_1,\dots,c_N) \in \Fb_q^N$ of as small a length $N$ as possible to be broadcast from the transmitter such that each user $u_i$, $i \in [m]$, can correctly estimate its own demanded message $\xb_{\xcal_i}$ using the codeword ${\bf{c}}$ and the noisy side information $\xb_{\xcal_i}^{e}$. Note that the task of decoding $\xb_{\xcal_i}$ from $\cb$ and $\xb_{\xcal_i}^{e}$ is equivalent to that of decoding the error vector $\eb_i = \xb_{\xcal_i}^{e} - \xb_{\xcal_i}$ at the user $u_i$.

The problem of designing a coding scheme for broadcasting $n$ information symbols to $m$ users with demands \mbox{$\xcal=(\xcal_1,\dots,\xcal_m)$} that are aided with noisy side information with at the most $\delta_s$ errors will be called the \emph{$(m,n,\xcal,\delta_s)$-BNSI (Broadcasting with Noisy Side Information) problem}. 
\begin{mydef}
A \emph{valid encoding function} of codelength $N$ for the ($m,n,\xcal,\delta_s$)-BNSI problem over the field $\Fb_q$ is a function
\begin{equation*}
\mathfrak{E}:\Fb_q^n\rightarrow \Fb_q^N
\end{equation*}
such that for each user $u_i,$ \mbox{$i \in [m]$} there exists a decoding function
\mbox{$\mathfrak{D}_i:\Fb_q^N \times \Fb_q^{|\xcal_i|} \rightarrow \Fb_q^{|\xcal_i|}$}
satisfying the following property:
$\mathfrak{D}_i(\mathfrak{E}(\xb),\xb_{\xcal_i}+\ei)=\xb_{\xcal_i}$
for every \mbox{$\xb \in \Fb_q^n$} and all \mbox{$\ei \in \Fb_q^{|\xcal_i|}$} with $wt(\ei)\leq \delta_s$. 
\end{mydef}
The aim of the code construction is to design a tuple $(\mathfrak{E},\mathfrak{D}_1,\mathfrak{D}_2,\dots, \mathfrak{D}_m)$ of encoding and decoding functions that minimizes the codelength $N$ and to calculate the \emph{optimal codelength} for the given problem which is the minimum codelength among all valid BNSI coding schemes. 
In this paper we will consider only linear coding schemes for the BNSI problem. By imposing linearity, we are able to utilize the rich set of mathematical tools available from linear algebra and the theory of error correcting codes to analyze the BNSI network. 

\begin{mydef}
A coding scheme $(\mathfrak{E},\mathfrak{D}_1,\mathfrak{D}_2,\dots, \mathfrak{D}_m)$ is said to be \emph{linear} if the encoding function \mbox{$\mathfrak{E}: \Fb_q^n \to \Fb_q^N$} is an $\Fb_q$-linear transformation. 
\end{mydef}

For a linear coding scheme, the codeword $\bf{c}=\mathfrak{E}(\xb)=\xb\Lb$, where $\xb \in \Fb_q^n$ and $\Lb \in \Fb_q^{n \times N}$. The matrix $\Lb$ is the \emph{encoder matrix} of the linear coding scheme. The minimum codelength among all valid linear coding schemes for the $(m,n,\xcal,\delta_s)$-BNSI problem over the field $\Fb_q$ will be denoted as either $N_{q,opt}(m,n,\xcal,\delta_s)$ or simply $N_{q,opt}$ if there is no ambiguity.

Note that the trivial coding scheme that transmits the information symbols $\xb$ `uncoded', i.e., $\bf{c}=\mathfrak{E}(\xb)=\xb$ is a valid coding scheme since each receiver $u_i$ can retrieve the demanded message $\xb_{\xcal_i}$ directly from the received codeword. Further, this code is linear with ${\bf{L}}={\bf{I}}_{n}$. Thus, we have the following trivial upper bound on the optimum linear codelength
\begin{equation} \label{eq:trivial}
N_{q,opt}(m,n,\xcal,\delta_s) \leq n.
\end{equation} 
We now introduce a representation of the BNSI problem as a bipartite graph.
\begin{mydef}
The bipartite graph \mbox{$\mathcal{B}=(\mathcal{U},\mathcal{P},\mathcal{E})$} corresponding to the $(m,n,\xcal,\delta_s)$-BNSI problem consists of the node-sets \mbox{$\mathcal{U}=\{u_1,u_2,\dots,u_m\}$} and \mbox{$\mathcal{P}=\{x_1,x_2,\dots,x_n\}$} and the set of undirected edges \mbox{$\mathcal{E}= \left\{\,\{u_i,x_j\}~|~i \in [m] \text{ and } j \in \xcal_i\right\}$}. 
\end{mydef}

The set $\mathcal{U}$ denotes the \textit{user-set} and $\mathcal{P}$ denotes the set of packets or the \textit{information symbol-set} and $\mathcal{E}$ represents the demands of each user in the broadcast channel.
Note that the degree of the user node $u_i$ in $\bcal$ equals $|\xcal_i|$.


\begin{figure}[h!]
\centering
\includegraphics[width=1.5in]{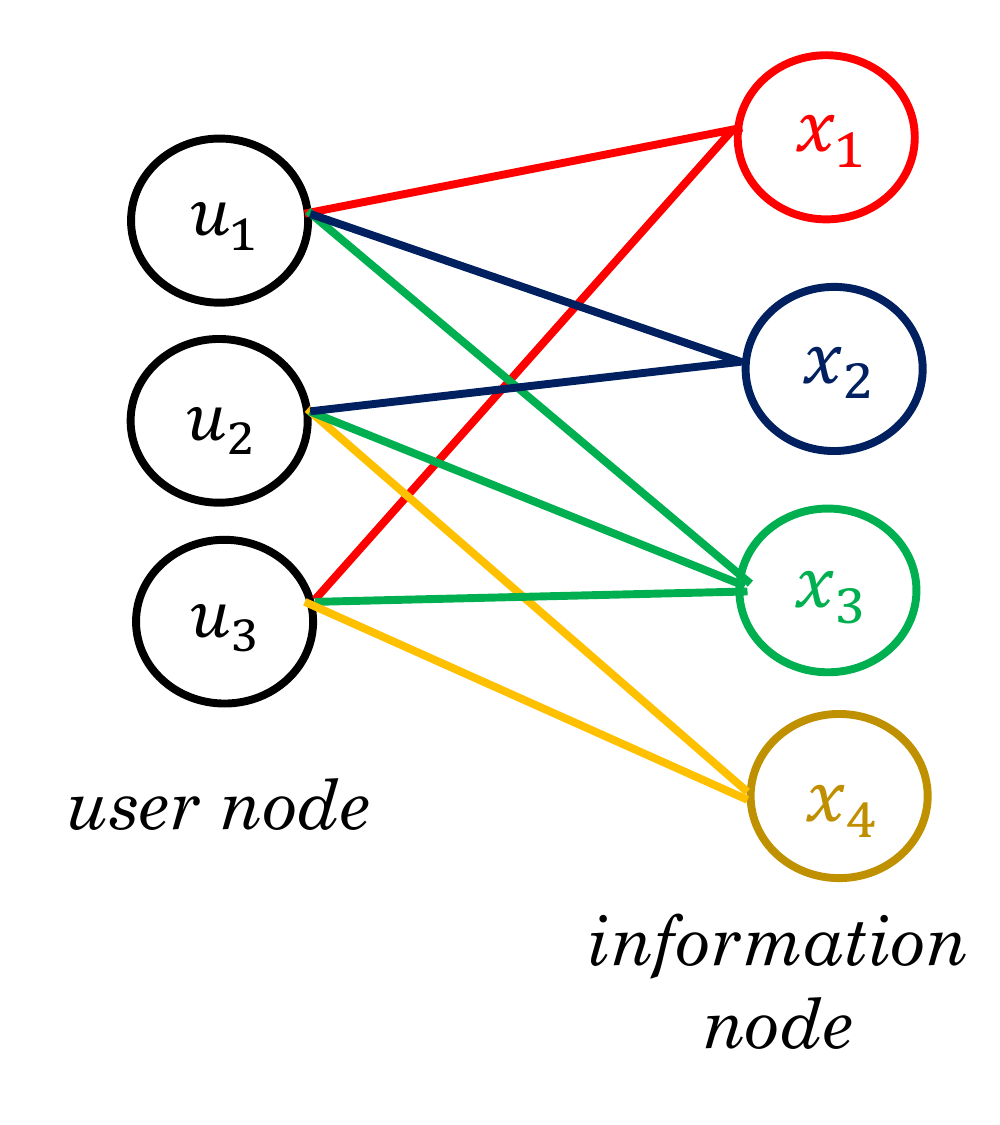}
\captionsetup{justification=centering}
\caption{\footnotesize{Bipartite graph $\mathcal{B}$ for the BNSI problem in Example~\ref{exmp1}}.}
\label{fig:image1}
\vspace{-5mm}
\end{figure} 

\begin{example} \label{exmp1}
Consider the BNSI problem with \mbox{$n=4$} information symbols, \mbox{$m=3$} users, and user demand index sets \mbox{$\xcal_1=\{1,2,3\}$}, \mbox{$\xcal_2=\{2,3,4\}$}, \mbox{$\xcal_3=\{1,3,4\}$}. The bipartite graph \mbox{$\mathcal{B}=(\mathcal{U},\mathcal{P},\mathcal{E})$} in Fig.~\ref{fig:image1} describes this scenario where \mbox{$\mathcal{U}=\{u_1,u_2,u_3\}$}, \mbox{$\mathcal{P}=\{x_1,x_2,x_3,x_4\}$} and $\mathcal{E}$= $\{\{u_1,x_1\}$,$\{u_1,x_2\}$,$\{u_1,x_3\}$,$\{u_2,x_2\}$,\\
$\{u_2,x_3\}$,$\{u_2,x_4\}$,$\{u_3,x_1\}$,$\{u_3,x_3\}$,$\{u_3,x_4\}\}$
\end{example}

\section{Design Criterion for the Encoder Matrix} \label{sec3}

Here, we derive a necessary and sufficient condition for a matrix ${\bf L} \in \Fb_q^{n \times N}$ to be a valid encoder matrix for the $(m,n,\xcal,\delta_s)$-BNSI problem over $\Fb_q$.
We now define the set $\ical(q,m,n,\xcal,\delta_s)$ of vectors ${\bf z}$ of length $n$ such that \mbox{$wt({\bf z}_{\xcal_i}) \in [2\delta_s]$} for some choice of $i \in [m]$, i.e., 
\begin{equation} \label{ical_def}
\ical(q,m,n,\xcal,\delta_s)=\bigcup\limits_{i=1}^{m} \left\{\,\zb \in \Fb_q^n~|~ 1 \leq wt(\zb_{\xcal_i}) \leq 2\delta_s \, \right\}.
\end{equation}  
When there is no ambiguity we will denote $\ical(q,m,n,\xcal,\delta_s)$ simply as $\ical$.
 
\begin{theorem} \label{thm1}
A matrix $\Lb\in\Fb_q^{n \times N}$ is a valid encoder matrix for the $(m,n,\xcal,\delta_s)$-BNSI problem if and only if
\begin{equation*}
{\bf zL} \neq {\bf 0},\quad \forall \zb \in \ical(q,m,n,\xcal,\delta_s).
\end{equation*} 
\end{theorem}
\begin{proof}
The encoding function $\mathfrak{E}(\xb)={\bf xL}$ is valid for the given BNSI problem if and only if for each $i \in [m]$, user $u_i$ can uniquely determine ${\bf x}_{\xcal_i}$ from the received codeword ${\bf xL}$ and the side information $\xb_{\xcal_i} + \eb_i$. Hence, for two distinct values of the demanded message $\xb_{\xcal_i}$ and $\xb'_{\xcal_i}$, if the noise vectors $\eb_i$ and $\eb'_i$ are such that the noisy side information at $u_i$ is identical (i.e., $\xb_{\xcal_i} + \eb_i = \xb'_{\xcal_i} + \eb'_i$) then the corresponding transmit codewords ${\bf xL}$ and ${\bf x'L}$ must be distinct for $u_i$ to distinguish the message ${\bf x}_{\xcal_i}$ from ${\xb'_{\xcal_i}}$.
Equivalently, the condition \mbox{${\bf xL} \neq {\bf x'L}$} should hold for every pair $\xb,\xb' \in \Fb_q^n$ such that \mbox{$\xb_{\xcal_i} \neq \xb'_{\xcal_i} $} and $\xb_{\xcal_i}+\ei=\xb'_{\xcal_i}+\ei'$ for some choice of $\ei,\ei' \in \Fb_q^{|\xcal_i|}$ with $wt(\ei), wt(\ei') \leq \delta_s$. 
Therefore, $\Lb$ is a valid encoder matrix if and only if 
\begin{equation}
\xb\Lb \neq \xb'\Lb \label{equ1}
\end{equation}
\mbox{$\forall \xb,\xb' \in \Fb_q^n$} such that \mbox{$\xb_{\xcal_i} \neq \xb'_{\xcal_i}$} and \mbox{$\xb_{\xcal_i}-\xb'_{\xcal_i}=\ei'-\ei$}, \mbox{$wt(\ei) \leq \delta_s$}, \mbox{$wt(\ei') \leq \delta_s$} for some $i \in [m]$ and $\ei,\ei' \in \Fb_q^{|\xcal_i|}$.  
Denoting \mbox{$\zb=\xb-\xb'$} the condition in~(\ref{equ1}) can be reformulated as \mbox{$\zb\Lb \neq {\pmb{0}}$} for all $\zb \in \Fb_q^n$ such that \mbox{$\zb_{\xcal_i} \neq \pmb{0}$} and \mbox{$\zb_{\xcal_i}=\ei'-\ei$}, \mbox{$wt(\ei) \leq \delta_s$}, \mbox{$wt(\ei') \leq \delta_s$}, for some choice of \mbox{$i \in [m]$} and $\ei,\ei' \in \Fb_q^{|\xcal_i|}$.
Equivalently, \mbox{${\bf zL} \neq {\bf 0}$} when \mbox{$wt(\zb_{\xcal_i})=wt(\ei-\ei')\leq 2\delta_s$} and \mbox{$wt(\zb_{\xcal_i}) \neq {\bf 0}$} for some \mbox{$i \in [m]$}. 
The statement of the theorem then follows.
\end{proof}

\begin{example} \label{exmp2}
Consider the $(3,4,\xcal,1)$-BNSI problem of Example~\ref{exmp1} with the field size $q=2$. It is straightforward to verify that $\mathcal{I}=\Fb_2^4\setminus\{\pmb{0},\pmb{1}\}$, where $\pmb{0}$ and $\pmb{1}$ denote the all zero and all one vector respectively in $\Fb_2^4$.
Let
\begin{equation} \label{eq:exmp2}
\Lb=
\begin{bmatrix}
1 & 0 & 0\\
0 & 1 & 0\\
0 & 0 & 1\\
1 & 1 & 1
\end{bmatrix}.
\end{equation}
It is easy to check that $\forall \zb \in \mathcal{I}$, $\zb\Lb \neq {\bf 0}$ because $wt(\zb)$ is either $1,2$ or $3$ and any $3$ rows of $\Lb$ are linear independent. 
Hence, the matrix ${\bf L}$ in~\eqref{eq:exmp2} is a valid encoder matrix, and this coding scheme with codelength $N=3$ saves $1$ channel use with respect to uncoded transmission.
It can be verified that no $4 \times 2$ binary matrix satisfies the criteria of Theorem~\ref{thm1} for this problem, and hence, $N_{2,opt}=3$.
\end{example}

We now provide a restatement of Theorem~\ref{thm1} in terms of the span of the rows of submatrices of ${\bf L}$. 
Towards this we first introduce some notation. For each $i \in [m]$, let \mbox{$\ycal_i=[n] \setminus \xcal_i$}. The set $\ycal_i$ is the index set of messages that are not demanded by $u_i$. 
For any $\acal \subseteq [n]$, ${\bf L}_{\acal}$ is the matrix consisting of the rows of $\Lb$ with indices given in $\acal$.

\begin{corollary} \label{corr1}
$\Lb$ is a valid encoder matrix for the $(m,n,\xcal,\delta_s)$-BNSI problem if and only if for every \mbox{$i \in [m]$},
any non-zero linear combination of any $2\delta_s$ or fewer rows of ${\bf{L}}_{\mathcal{X}_i}$ does not belong to $rowspan\{\Lb_{\mathcal{Y}_i}\}$.
\end{corollary}
\begin{proof} 
From Theorem~\ref{thm1}, ${\bf L}$ is a valid encoder matrix if and only if ${\bf zL} \neq 0$ whenever $1 \leq wt({\bf z}_{\xcal_i}) \leq 2\delta_s$ for some choice of $i \in [m]$. Since ${\bf zL} = \zb_{\xcal_i}\Lb_{\xcal_i} + \zb_{\ycal_i}\Lb_{\ycal_i}$, we have $\zb_{\xcal_i}\Lb_{\xcal_i} \neq -\zb_{\ycal_i}\Lb_{\ycal_i}$ for any \mbox{$\zb_{\xcal_i} \in \Fb_q^{|\xcal_i|} \setminus \left\{ {\bf 0} \right\}$} with $wt(\zb_{\xcal_i}) \leq 2\delta_s$ and for any \mbox{$\zb_{\ycal_i} \in \Fb_q^{|\ycal_i|}$}. The corollary then immediately follows from this observation.
\end{proof}

\begin{corollary} \label{corr2} 
If $\Lb$ is a valid encoder matrix for the $(m,n,\xcal,\delta_s)$-BNSI problem then any $2\delta_s$ or fewer rows of $\Lb_{\xcal_i}$ are linearly independent for every $i \in [m]$.
\end{corollary}
\begin{proof}
From Corollary~\ref{corr1}, if ${\bf L}$ is a valid encoder matrix any non-zero linear combination of $2\delta_s$ or fewer rows of ${\bf L}_{\xcal_i}$ is not in $rowspan\{\Lb_{\ycal_i}\}$. Since ${\bf 0} \in rowspan\{\Lb_{\ycal_i}\}$, $\zb_{\xcal_i}\Lb_{\xcal_i} \neq {\bf 0}$ if $1 \leq wt(\zb_{\xcal_i}) \leq 2\delta_s$.
Therefore, any $2\delta_s$ or fewer rows of $\Lb_{\xcal_i}$ must be linearly independent.
\end{proof}

\section{Syndrome Decoding} \label{sec4}

We now propose a decoding procedure for linear coding schemes for an arbitrary $(m,n,\xcal,\delta_s)$-BNSI problem which uses similar concept of \textit{syndrome decoding} for linear error correcting codes. Consider a code for $(m,n,\xcal,\delta_s)$-BNSI problem generated by a valid encoder matrix $\Lb \in \Fb_q^{n \times N}$. 
The user $u_i, i \in [m]$ receives the codeword 
$\xb\Lb=\xb_{\xcal_i}\Lb_{\xcal_i}+\xb_{\ycal_i}\Lb_{\ycal_i}$
and also possesses erroneous demanded information symbol vector $\xb_{\xcal_i}^{e}=\xb_{\xcal_i}+ \ei$, where $\ei \in \Fb_q^{|\xcal_i|}$ and $wt(\ei)\leq \delta_s$.
 
Considering user $u_i$, suppose \mbox{$\beta_i \subseteq \ycal_i$} denotes the index set of the rows of $\Lb$ that form a basis for $rowspan\{\Lb_{\ycal_i}\}$. Therefore $rowspan\{\Lb_{\beta_i}\}$ = $rowspan\{\Lb_{\ycal_i}\}$, and $\Lb_{\beta_i}$ has linearly independent rows. So we can write $\xb_{\ycal_i}\Lb_{\ycal_i}=\bb\Lb_{\beta_i}$ for some  $\bb \in \Fb_q^{|\beta_i|}$. 
Hence, the received codeword is $\cb = \xb\Lb=\xb_{\xcal_i}\Lb_{\xcal_i} + \bb\Lb_{\beta_i}$. Note that $\bb\Lb_{\beta_i}$ is the interference at receiver $u_i$ due to the undesired messages $\xb_{\ycal_i}$.
Regarding $rowspan\{\Lb_{\beta_i}\}$ as a linear code of length $N$ and dimension $|\beta_i|$ over $\Fb_q$, let \mbox{$\Hb_i \in \Fb_q^{(N-|\beta_i|) \times N}$} be a parity check matrix of $rowspan\{\Lb_{\beta_i}\}$. Since $\Lb_{\beta_i}$ is a generator matrix of this code, we have $\Hb_i\Lb^T_{\beta_i}={\bf 0}$.

The syndrome decoder at $u_i$ functions as follows. Given the codeword $\cb$ and the noisy side information $\xb_{\xcal_i}^{e}$, the receiver first computes 
\begin{align*}
\yb' &= \xb_{\xcal_i}^{e}\Lb_{\xcal_i} - \cb     
= (\xb_{\xcal_i}+ \ei)\Lb_{\xcal_i} - (\xb_{\xcal_i}\Lb_{\xcal_i}+ \bb\Lb_{\beta_i})\\ 
&= \ei\Lb_{\xcal_i} - \bb\Lb_{\beta_i}.   
\end{align*}  
In order to remove the interference from $\bb$, the receiver multiplies ${\yb'}^T$ with $\Hb_i$ to obtain the \emph{syndrome}
\begin{equation*}
\pmb{b}_i^T = \Hb_i {\yb'}^T= \Hb_i {\Lb_{\xcal_i}}^T {\ei}^T - \Hb_i\Lb^T_{\beta_i} \bb^T = \Hb_i {\Lb_{\xcal_i}}^T {\ei}^T.
\end{equation*} 
Defining \mbox{$\Ab_i=\Hb_i {\Lb_{\xcal_i}}^T$}, we have $\Ab_i \ei^T=\pmb{b}^T_i$. Given the syndrome $\pmb{b}_i$ and the matrix $\Ab_i$, the receiver must identify the error vector $\ei$. We now show that $\pmb{b}_i=\Ab_i \ei^T$ uniquely determines $\ei$ provided $wt(\ei) \leq \delta_s$.

\begin{lemma}  \label{lmm11}
If $\Lb$ is a valid encoder matrix and $\ei$ and $\ei'$ are distinct vectors in $\Fb_q^{|\xcal_i|}$ each with Hamming weight at the most $\delta_s$, then $\Ab_i\ei^T \neq \Ab_i\ei'^T$.
\end{lemma}
\begin{proof}
Proof by contradiction.
Suppose $\exists$ $\ei,\ei' \in \Fb_q^{|\xcal_i|}$ such that \mbox{$\ei \neq \ei'$} and \mbox{$wt(\ei) \leq \delta_s$}, \mbox{$wt(\ei')\leq \delta_s$} which satisfies \mbox{$\Ab_i \ei^T=\Ab_i \ei'^T$}.
Then we have
\begin{align*}
&\Hb_i {\Lb_{\xcal_i}}^T \ei^T = \Hb_i {\Lb_{\xcal_i}}^T \ei'^T\\
\Rightarrow ~&\Hb_i ((\ei-\ei')\Lb_{\xcal_i})^T = \pmb{0}\\
\Rightarrow ~&(\ei-\ei')\Lb_{\xcal_i} \in rowspan\{\Lb_{\beta_i}\} = rowspan\{\Lb_{\ycal_i}\}.
\end{align*}
Assuming \mbox{$\ei-\ei'=\ei''$}, we have \mbox{$1 \leq wt(\ei'') \leq 2\delta_s$} and \mbox{$\ei''\Lb_{\xcal_i} \in rowspan\{\Lb_{\ycal_i}\}$}. This implies that there exists a non-zero linear combination of $2\delta_s$ or fewer rows of $\Lb_{\xcal_i}$ that belongs to $rowspan\{\Lb_{\ycal_i}\}$ which contradicts the necessary and sufficient criterion (Corollary~\ref{corr1}) for $\Lb$ to be a valid encoder matrix. 
\end{proof}

Now Lemma~\ref{lmm11} leads us to the following syndrome decoding procedure:
given the received codeword $\cb$ and side information $\xb_{\xcal_{i}}^{e}$, the receiver first computes the syndrome $\pmb{b}_i^T = \Hb_i (\xb_{\xcal_i}^e\Lb_{\xcal_i} -\cb )^T$, and then identifies, either by exhaustive search or by using a look up table, the unique vector $\ebh \in \Fb_q^{|\xcal_i|}$ of weight at the most $\delta_s$ that satisfies $\Ab_i \ebh^T = \pmb{b}_i^T$.
If the Hamming weight of the noise $\eb_i$ is at the most $\delta_s$, then the estimate $\ebh$ equals $\ei$, and the receiver retrieves the demanded message through $\xb_{\xcal_i} = \xb_{\xcal_i}^{e} - \ebh$. The algorithm for \emph{syndrome decoding} is given in Algorithm~1 which is valid for any $i \in [m]$. 
Similar to the syndrome decoding procedure of a general linear error correcting code, the proposed algorithm relies on an exhaustive search (or a look up table) to identify the unique solution of weight at the most $\delta_s$ to the linear equation $\Ab_i \ebh = \pmb{b}_i^T$. We are yet to address the problem of designing coding schemes that admit efficient low-complexity implementations of syndrome decoding. 

 \begin{algorithm}[t]
 \caption{Syndrome Decoding}
 \SetAlgoLined
 \textbf{Input: $\cb$, $\xb_{\xcal_i}^{e}$, $\Lb$, $\Hb_i$, $\Ab_i$}\\
 \textbf{Output}: An estimate $\ebh$ of the error vector $\ei$\\
 \textbf{Procedure}
  \begin{itemize}
  \item[] Step 1: Compute $\yb'= (\cb-\xb_{\xcal_i}^e\Lb_{\xcal_i})$ 
 \item[] Step 2: Compute syndrome $\pmb{b}_i^T = \Hb_i \yb'^T $
  \item[] Step 3: Calculate $\Ab_i \eb^T, \forall \eb \in \Fb_q^{|\xcal_i|}$ with $wt(\eb) \leq \delta_s$. Among these vectors identify a vector $\ebh$ that satisfies $\Ab_i \ebh^T=\pmb{b}_i^T$   
 \end{itemize}
\end{algorithm}

\begin{example} \label{exmp3}
We now consider syndrome decoding at user $u_1$ (i.e., $i=1$) for the BNSI problem of Example~\ref{exmp1} with the binary ($q=2$) encoder matrix $\Lb$ given in~\eqref{eq:exmp2} in Example~\ref{exmp2}. For $u_1$, we have $\xcal_1=\{1,2,3\}$, $\ycal_1=\{4\}$, $\Lb_{\xcal_1}={\bf I}_3$ and $\Lb_{\ycal_1}=(1~1~1)$. In this case, the rows indexed by $\beta_1=\ycal_1=\{4\}$ form a basis for $rowspan\{\Lb_{\ycal_1}\}$.
A parity check matrix for $rowspan\{\Lb_{\beta_1}\}$ is
\begin{equation*}
\Hb_1=
\begin{bmatrix}
1 & 0 & 1 \\
0 & 1 & 1
\end{bmatrix}.
\end{equation*}
The corresponding $\Ab_1$ matrix is $\Ab_1= \Hb_1 \Lb_{\xcal_1}^T = \Hb_1 \, {\bf I}_3 = \Hb_1$. The value of $\Ab_1 \eb^T$ for all possible $\eb$ of weight at the most $\delta_s=1$ is given in the following look up table.

\begin{center}
\renewcommand{\arraystretch}{1.25}
\begin{tabular}{c|c|c|c|c}
\hline
$\eb$ & (0 0 0) & (0 0 1) & (0 1 0) & (1 0 0) \\
\hline
$\Ab_1 \eb^T$ & (0 0)$^T$ & (1 1)$^T$ & (0 1)$^T$ & (1 0)$^T$ \\
\hline
\end{tabular}
\end{center}
Note that the syndrome $\Ab_1\eb^T$ is distinct for each possible error vector $\eb$.

Suppose $\xb= (1~0~0~1)$, i.e., the message vector demanded by $u_1$ is $\xb_{\xcal_1}=(1~0~0)$. The transmitter will transmit the codeword $\cb=\xb\Lb=$ (0 1 1). Suppose user $u_1$ has the erroneous demanded information symbol vector \mbox{$\xb_{\xcal_1}^{e}=$ (1 0 1)}, i.e., \mbox{$\eb_1=$ (0 0 1)}. User $u_1$ will calculate the syndrome $\pmb{b}_1^T=\Hb_1(\xb_{\xcal_1}^{e}\Lb_{\xcal_1} - \cb)^T = (1~1)^T$. Using the syndrome look up table, the decoder will output $\ebh=(0~0~1)$ as the estimated error vector. This is subtracted from $\xb_{\xcal_i}^{e}=(1~0~1)$ to obtain the estimate 
$(1~0~0)$ of the demanded message $\xb_{\xcal_1}$.
\end{example}

\section{Characterization of Networks with $N_{q,opt} < n$} \label{sec5}

We remarked in Section~\ref{sec2} that uncoded transmission $\Lb={\bf I}_n$ is a valid linear coding scheme where number of channel uses $N$ is equal to the length $n$ of the message vector. It is important to identify the subset of BNSI problems for which this uncoded transmission is optimal (i.e., $N_{q,opt}=n$), or equivalently, characterize the family of networks where linear coding provides strict gains over uncoded transmission (i.e., \mbox{$N_{q,opt} < n$}). 
This will allow us to identify the key structural properties of BNSI problems that lead to performance gains through network coding and will be helpful in conceiving systematic constructions of explicit encoder matrices.

\subsection{Preliminaries}

We now derive a few results based on which we formulate a necessary and sufficient condition for a BNSI problem to have $N_{q,opt}=n$.  

\begin{lemma} \label{lemma1}
If $\Lb$ is a valid encoder matrix for an $(m,n,\xcal,\delta_s)$-BNSI problem and if \mbox{$|\xcal_i| \leq 2\delta_s$}, for some \mbox{$i \in [m]$}, then \mbox{$rank(\Lb_{\xcal_i})=|\xcal_i|$}, $rowspan(\Lb_{\xcal_i}) \cap rowspan(\Lb_{\ycal_i}) =\{{\bf 0}\}$, and
$rank(\Lb) = rank(\Lb_{\xcal_i}) + rank(\Lb_{\ycal_i}) = |\xcal_i| + rank(\Lb_{\ycal_i})$.
\end{lemma}
\begin{proof}
Follows immediately from Corollaries~\ref{corr1} and~\ref{corr2} using the fact that number of rows of $\Lb_{\xcal_i}$ is not more than $2\delta_s$, and using the observation that the rows of $\Lb_{\xcal_i}$ are linearly independent and their span intersects trivially with the span of the rows of $\Lb_{\ycal_i}$.
\end{proof}

From Lemma~\ref{lemma1}, if $|\xcal_i| \leq 2\delta_s$, the rows of $\Lb$ corresponding to the message vectors $\xb_{\xcal_i}$ and those corresponding to $\xb_{\ycal_i}$ are linearly independent. Hence, when encoded using $\Lb$ the message symbols $\xb_{\ycal_i}$ do not interfere with the detection of the symbols in $\xb_{\xcal_i}$.

\subsubsection{Subproblems of a given BNSI problem} \label{sec5_subproblem}

Let $(m,n,\xcal,\delta_s)$ be any given BNSI problem.
Consider the $(m',n',\xcal',\delta_s)$-BNSI problem derived from $(m,n,\xcal,\delta_s)$ by removing the symbols $\xb_{\xcal_i}$, for some $i \in [m]$, from the demands of all the receivers. The derived problem has \mbox{$m'=m-1$} users one corresponding to each $j \in [m] \setminus \{i\}$. The demand of the $j^{\text{th}}$ user is $\xcal'_j = \xcal_j \setminus \xcal_i = \xcal_j \cap \ycal_i$, and $\xcal'=(\xcal'_j, j \in [m] \setminus \{i\})$.
The vector of information symbols for the new problem is $\xb_{[n]\setminus \xcal_i} = \xb_{\ycal_i}$, and the number of message symbols is \mbox{$n'=n-|\xcal_i|$}. The bipartite graph $\bcal'=(\ucal',\pcal',\ecal')$ for the derived problem will consist of the user-set $\ucal'=\ucal \setminus \{u_i\}$, information symbol-set $\pcal'=\{x_k | k \notin \xcal_i\}$, and edge set $\ecal' = \{ \, \{u_j,x_k\} \in \ecal \,| \, k \notin \xcal_i \,\}$. Note that $\bcal'$ is the subgraph of $\bcal$ induced by the nodes $\{x_k | k \notin \xcal_i\}$. 

\begin{lemma} \label{lemma_subproblem}
If $\Lb$ is a valid encoder matrix for $(m,n,\xcal,\delta_s)$, then $\Lb_{\ycal_i}$ is a valid encoder matrix for $(m',n',\xcal',\delta_s)$.
\end{lemma}
\begin{proof}
For any $j \in [m] \setminus \{i\}$ we have, \mbox{$\ycal'_j = \ycal_j \cap \ycal_i \subset \ycal_j$} and $\xcal'_j = \xcal_j \cap \ycal_i \subset \xcal_j$. From Corollary~\ref{corr1}, any non-zero linear combination of $2\delta_s$ or fewer rows of $\Lb_{\xcal_j}$ does not belong to $rowspan\{\Lb_{\ycal_j}\}$.
Since $rowspan\{\Lb_{\ycal'_j}\} \subset rowspan\{\Lb_{\ycal_j}\}$ and $\Lb_{\xcal'_j}$ is a submatrix of $\Lb_{\xcal_j}$, we deduce that any non-zero linear combination of $2\delta_s$ or fewer rows of $\Lb_{\xcal'_j}$ is not in $rowspan\{\Lb_{\ycal'_j}\}$. Lemma~\ref{lemma_subproblem} then follows from Corollary~\ref{corr1}.\end{proof}

\subsubsection{A simple coding scheme for a family of BNSI problems} \label{sec:simple_coding}

Consider any BNSI problem $(m,n,\xcal,\delta_s)$ where $|\xcal_i| \geq 2\delta_s + 1$ for all $i \in [m]$, i.e., $|\ycal_i| \leq n - 2\delta_s - 1$. We will now provide a simple coding scheme with $N=n-1$ for any such problem. Let $\Lb \in \Fb_q^{n \times (n-1)}$ be such that its first $(n-1)$ rows form the identity matrix ${\bf I}_{n-1}$ and the last row is the all-one vector ${\bf 1}=(1~1~\cdots~1) \in \Fb_q^{(n-1)}$. Observe that any $(n-1)$ rows of $\Lb$ are linearly independent. We now show that $\Lb$ satisfies the condition in Theorem~\ref{thm1}. For any $\zb \in \ical$, there exists an $i \in [m]$ such that $wt(\zb_{\xcal_i}) \leq 2\delta_s$. 
Using \mbox{$|\ycal_i| \leq n-2\delta_s-1$},
\begin{align*}
wt(\zb) &= wt(\zb_{\xcal_i}) + wt(\zb_{\ycal_i}) 
        \leq 2\delta_s + n - 2\delta_s - 1 = n-1.
\end{align*} 
Since any $(n-1)$ rows of $\Lb$ are linearly independent, $\zb\Lb \neq {\bf 0}$. This proves that $\Lb$ is a valid encoder matrix for this problem.
We do not claim that this scheme is optimal, however, this scheme is useful in proving the main result of this section. The linear code in Example~\ref{exmp2} is an instance of this coding scheme.

\subsection{Characterization of networks with $N_{q,opt} < n$} \label{sec5b}

Suppose a bipartite graph \mbox{$\mathcal{B}=(\mathcal{U},\mathcal{P},\mathcal{E})$} represents an $(m,n,\xcal,\delta_s)$-BNSI problem.
We now define a collection $\Phi(\bcal)$ of subsets of information symbol indices. A non-empty set $\cs \subset [n]$ belongs to $\Phi(\bcal)$ if and only if the subgraph $\bcal'=(\ucal',\pcal',\ecal')$ of $\bcal$ induced by the packet nodes $\pcal_{\cs}=\{x_k \, | \, k \in \cs\}$ has the following property:
$deg(u)\geq 2\delta_s+1$ for all $u \in \ucal'$,  
where $deg(u)$ is the number of edges incident on the vertex $u$. 
Equivalently, $\Phi(\mathcal{B})$ is the collection of all non-empty $\cs \subset [n]$ such that 
\begin{equation*}
\text{for every } i \in [m],~|\xcal_i \cap \cs| \notin [2\delta_s],\text{i.e., either}~|\xcal_i \cap \cs|=0~\text{or}~|\xcal_i \cap \cs|\geq 2\delta_s+1.
\end{equation*} 

\begin{lemma} \label{lem:suff}
If $\Phi(\bcal)$ is empty, i.e, $\Phi(\bcal)=\phi$, then $N_{q,opt}=n$.
\end{lemma}
\begin{proof}
Let $\Lb$ be an optimal encoder matrix with \mbox{$N=N_{q,opt}$}. Since $\Phi(\mathcal{B})=\phi$, there doesn't exist any non-empty $\cs \subset [n]$ such that \mbox{$|\xcal_i \cap \cs| \notin [2\delta_s],~\forall i \in [m]$}. In particular, choosing \mbox{$\cs=[n]$} we deduce that there exists at least one user $u_{i_1}$, $i_1 \in [m]$ such that \mbox{$1 \leq |\xcal_{i_1}| \leq 2\delta_s$}. 
By Lemma~\ref{lemma1}, $rank(\Lb)=rank(\Lb_{\xcal_{i_1}}) + rank(\Lb_{\ycal_{i_1}}) = |\xcal_{i_1}| + rank(\ycal_{i_1})$.
Removing the information symbols $\xb_{\xcal_{i_1}}$ from the problem $(m,n,\xcal,\delta_s)$, we obtain a derived BNSI problem $(m^{(1)},n^{(1)},\xcal^{(1)},\delta_s)$ (see Lemma~\ref{lemma_subproblem}, Section~\ref{sec5_subproblem}) where $m^{(1)}=m-1$, $n^{(1)}=n-|\xcal_{i_1}|$ and $\xcal^{(1)}=(\xcal_j \cap \ycal_{i_1}, j \neq i_1)$. 
From Lemma~\ref{lemma_subproblem}, the matrix $\Lb^{(1)}=\Lb_{\ycal_{i_1}}$ is a valid encoder for this problem.
The bipartite graph $\bcal^{(1)}$ of the derived problem is a subgraph of $\bcal$. Since $\Phi(\bcal)$ is empty, it follows from the definition of $\Phi$ that $\Phi(\bcal^{(1)})$ is empty as well. Also, $rank(\Lb) = |\xcal_{i_1}| + rank(\Lb^{(1)})$.

Since $\Phi(\bcal^{(1)})$ is empty, the arguments used with the original problem $\bcal$ in the previous paragraph hold for the derived problem $\bcal^{(1)}$ as well. Hence, there exists an $i_2 \in [m] \setminus \{i_1\}$ such that $rank(\Lb^{(1)})$ = \mbox{$|\xcal_{i_2} \setminus \xcal_{i_1}|$} + $rank(\Lb_{\ycal_{i_1} \cap \ycal_{i_2}})$, and $\Lb^{(2)}=\Lb_{\ycal_{i_1} \cap \ycal_{i_2}}$ is a valid encoder matrix for the problem $(m^{(2)},n^{(2)},\xcal^{(2)},\delta_s)$ derived from $(m^{(1)},n^{(1)},\xcal^{(1)},\delta_s)$ by removing $\xb_{\xcal_{i_2} \setminus \xcal_{i_1}}$.
The bipartite graph $\bcal^{(2)}$ for this problem is a subgraph of $\bcal^{(1)}$, and hence, satisfies $\Phi(\bcal^{(2)}) = \phi$.
Note that $rank(\Lb) = |\xcal_{i_1}| + rank(\Lb^{(1)}) = |\xcal_{i_1}| + |\xcal_{i_2} \setminus \xcal_{i_1}| + rank(\Lb^{(2)}) = |\xcal_{i_1} \cup \xcal_{i_2}| + rank(\Lb^{(2)})$.

We will continue this process until the size of the information symbols-set is at the most $2\delta_s$. Say this will happen in $t^{\text{th}}$ iteration. Then, the matrix $\Lb^{(t)}=\Lb_{\ycal_{i_1} \cap \cdots \cap \ycal_{i_t}}$ is a valid encoder matrix for the $t^{\text{th}}$ derived BNSI problem $(m^{(t)},n^{(t)},\xcal^{(t)},\delta_s)$, and $rank(\Lb) = |\xcal_{i_1} \cup \cdots \cup \xcal_{i_t}| + rank(\Lb^{(t)})$.
Since $\Lb^{(t)}$ has at the most $2\delta_s$ rows, from Corollary~\ref{corr1}, all the rows of $\Lb^{(t)}$ are linearly independent, and hence, $rank(\Lb^{(t)})=|\ycal_{i_1} \cap \cdots \cap \ycal_{i_t}| = n - |\xcal_{i_1} \cup \cdots \cup \xcal_{i_t}|$.
It then follows that $rank(\Lb)=n$.
Thus, the number of columns $N_{q,opt}$ of $\Lb$ satisfies $N_{q,opt} \geq rank(\Lb) \geq n$. From~\eqref{eq:trivial}, we have $N_{q,opt} \leq n$ thereby proving that $N_{q,opt}=n$.
\end{proof}

We will now show that $N_{q,opt}=n$ only if $\Phi(\bcal) = \phi$. 
\begin{lemma} \label{lem:necess}
If $\Phi(\bcal) \neq \phi$, then $N_{q,opt} < n$. 
\end{lemma}
\begin{proof}
Here we will provide a constructive proof where we will design a valid coding scheme with $N < n$.
Since $\Phi(\bcal)$ is non-empty, there exists a non-empty $\cs \subset [n]$ such that for each $i \in [m]$ either $|\xcal_i \cap \cs| \geq 2\delta_s + 1$ or $\xcal_i \cap \cs = \phi$. The proposed linear coding scheme partitions the transmit codeword $\cb$ into two parts $(\cb_1~\cb_2)$. The vector $\cb_1$ carries the symbols $\xb_{[n]\setminus \cs}$ uncoded, i.e., $\cb_1 = \xb_{[n] \setminus \cs}$. When $\cb_2$ is broadcast, we will assume all the receivers know the value of $\xb_{[n]\setminus \cs}$. Thus, the problem of designing the second part of the code transmission, wherein the symbols $\xb_{\cs}$ must be delivered to the receivers, is identical to the BNSI problem $\bcal'=(\ucal',\pcal',\ecal')$ with information symbol-set $\pcal'=\{ x_j \, | j \in \cs\}$, user-set $\ucal'=\{ u_i \, | \, \xcal_i \cap \cs \neq \phi\}$ and demands $\xcal' = (\xcal_i \cap \cs, \forall u_i \in \ucal')$. Since $\cs \in \Phi(\bcal)$, $\xcal_i \cap \cs \neq \phi$ implies $|\xcal_i \cap \cs| \geq 2\delta_s + 1$. Thus, the demand set of every receiver in the problem $\bcal'$ has cardinality at least $2\delta_s+1$. By using the coding scheme of Section~\ref{sec:simple_coding} for the problem $\bcal'$, we require a code length of $|\pcal'|-1=|\cs|-1$ for the vector $\cb_2$.
Hence, the codelength $N$ of the overall coding scheme is the sum of the lengths of $\cb_1$ and $\cb_2$, i.e., $N = n - |\cs| + |\cs| - 1 = n - 1$. We conclude that $N_{q,opt} < n$. 
\end{proof}

The main result of this section follows immediately from Lemmas~\ref{lem:suff} and~\ref{lem:necess}.

\begin{theorem} \label{thm2}
For an $(m,n,\xcal,\delta_s)$-BNSI problem represented by the bipartite graph $\bcal$, $N_{q,opt} = n$ if and only if $\Phi(\bcal) = \phi$.
\end{theorem}

\subsection{An algorithm to determine if $\Phi(\bcal)$ is empty} \label{sec5c}

We now propose a simple iterative procedure given in Algorithm~2 which determines whether $\Phi(\mathcal{B})=\phi$ for a given bipartite graph $\mathcal{B}=(\mathcal{U},\mathcal{P},\mathcal{E})$. The idea behind Algorithm~2 is to find $\mathsf{P}_{\cs} \subseteq \mathcal{P}$ for which each user-node in the subgraph induced by information symbol-set $\mathsf{P}_{\cs}$ has degree either $0$ or $2\delta_s+1$. The procedure in Algorithm~2 proceeds as follows 
\begin{itemize}
\item[$~$] Initialize $\mathsf{B}=(\mathsf{U},\mathsf{P},\mathsf{E})$, where $\mathsf{U}=\ucal$, $\mathsf{P}=\pcal$, $\mathsf{E}=\ecal$.
\item[1.] Check whether every user-node in $\mathsf{U}$ has degree at least $2\delta_s+1$ (It can not be $0$ because each user has non-empty demanded information symbol index set). If true, then $\{j \, | \, x_j \in \mathsf{P} \,\} \in \Phi(\mathcal{B})$ and $\Phi(\bcal)$ is non-empty. If false, proceed to Step~2.
\item[2.] Find a user-node $u_i$ with $1 \leq deg(u_i) \leq 2\delta_s$. Modify the graph $\mathsf{B}$ by removing the packet nodes $\{x_j \, | \, j \in \xcal_i\}$ and all the edges incident on these packet nodes. 
Then, remove any user node with zero degree.
If $|\mathsf{P}| \leq 2\delta_s$ declare $\Phi(\bcal)=\phi$, else go to Step~1. 
\end{itemize}
\begin{algorithm}[t]
\caption{Algorithm to determine if $\Phi(\mathcal{B})=\phi$}
\SetAlgoLined
\textbf{Input}: $\mathcal{B}=(\mathcal{U},\mathcal{P},\mathcal{E})$, $\delta_s$\\
\textbf{Output}: {\sf TRUE} if $\Phi(\mathcal{B})=\phi$, {\sf FALSE} otherwise and one element $\mathsf{C} \in \Phi(\bcal)$\\
 \% \% {Initialization:} \\
 $\mathsf{U} \leftarrow \mathcal{U}$,
 $\mathsf{P} \leftarrow \mathcal{P}$,
 $\mathsf{E} \leftarrow \mathcal{E}$,
 Bipartite graph $\mathsf{B}=(\mathsf{U},\mathsf{P},\mathsf{E})$ \\
 \% \% Iteration: \\
 \While{$|\mathsf{P}| > 2\delta_s $}{
   \eIf{$\forall u \in \mathsf{U},$ $deg(u)\geq 2\delta_s+1$}{
   $\mathsf{C} \leftarrow \{j~|~x_j \in \mathsf{P}\}$\\
   \texttt{output} {\sf FALSE}; \texttt{return};
   }{
   Find a $u_i \in \mathsf{U}$ such that $1 \leq deg(u_i) \leq 2\delta_s$ \\
   $\mathsf{P}$ $\leftarrow$ $\mathsf{P} \setminus \{ x_j \, | \, j  \in \xcal_i \}$ \\
   $\mathsf{E}$ $\leftarrow$ $\mathsf{E} \setminus \left\{ \, \{u_k,x_j\}\, |\,j \in \xcal_i \text{ and } \{u_k,x_j\} \in \mathsf{E} \right\}$ \\
   $\mathsf{U} \leftarrow \mathsf{U} \setminus \{ u_k \, | \, deg(u_k) = 0 \}$
  }
 }
\texttt{output} TRUE; \texttt{return};
\% \% $|\mathsf{P}| \leq 2\delta_s$, hence $\Phi =\phi$ 
\end{algorithm}

The correctness of the algorithm follows from the observation that the subgraph of $\mathsf{B}$ obtained in Step~2 by removing the packet nodes $\{x_j | j \in \xcal_i\}$ has non-empty $\Phi$ if and only if the set $\Phi(\mathsf{B})$ of the original graph $\mathsf{B}$ is itself non-empty. This is due to the fact that any member of $\Phi(\mathsf{B})$ will contain no elements from $\xcal_i$ since the degree of $u_i$ is at the most $2\delta_s$.

\begin{example}
Consider the scenario mentioned in Example \ref{exmp2}. Applying Algorithm 2 we will obtain \mbox{$\{1,2,3,4\} \in \Phi(\mathcal{B})$}, so \mbox{$N_{q,opt}<n$}. A valid encoding and decoding scheme over $\Fb_2$ with codelength $3$ for this scenario is given in Example \ref{exmp3}. If we consider the following scenario where \mbox{$n=5$}, \mbox{$m=4$}, \mbox{$\delta_s=1$}, \mbox{$\mathcal{X}_1=\{1,2,3,4\}$}, \mbox{$\mathcal{X}_2=\{4,5\}$}, \mbox{$\mathcal{X}_3=\{1,3,5\}$} and \mbox{$\mathcal{X}_4=\{1,2,4\}$}. Again applying Algorithm 1, we can conclude that for this scenario $\Phi(\mathcal{B})=\phi$, therefore $N_{q,opt}=n=5$
\end{example}
\section{Bounds on $N_{q,opt}$ and some code constructions}
Until now we have not described any systematic construction of an encoder matrix $\Lb$ or any methodology for calculating the optimal codelength $N_{q,opt}$ for a general $(m,n,\xcal,\delta_s)$ BNSI problem. In this section we will present some lower bounds on the optimal codelength $N_{q,opt}$ and constructions of encoder matrices $\Lb$ for $(m,n,\xcal,\delta_s)$ BNSI problem. These constructions will provide upper bounds on $N_{q,opt}$.
\subsection{Lower Bounds on $N_{q,opt}$}  \label{lb}
Here we will describe two lower bounds on $N_{q,opt}$, one of them is based on the size of the \emph{demanded information symbol index set} of each user in a given BNSI problem and the other will be characterized based on the set $\Phi$ defined on a subgraph of the bipartite graph representing the BNSI problem. At first, we will derive a result that will help to obtain the lower bounds on optimal codelength described in the two subsequent sub-sections.

Consider a bipartite graph $\mathcal{B}=(\mathcal{U},\mathcal{P},\mathcal{E})$ that represents the $(m,n,\xcal,\delta_s)$ BNSI problem. For any $\rho \subseteq [n]$, let $x_{\rho}=\{x_j~|~j \in \rho\}$. We will derive a subgraph $\mathcal{B'}=(\mathcal{U'},\mathcal{P'},\mathcal{E'})$ from $\mathcal{B}$ induced by the \emph{information set} \mbox{$x_{\rho}=\mathcal{P'} \subseteq \mathcal{P}$}, where \mbox{$\mathcal{U'}=\{u_i \in \mathcal{U}~|\rho \cap \xcal_i \neq \phi\}$} and $\mathcal{E'}=\{\{u_i,x_j\}\in \mathcal{E}~|~x_j \in 
\mathcal{P'}, u_i \in \mathcal{U'}\}$. The bipartite graph $\mathcal{B'}=(\mathcal{U'},\mathcal{P'},\mathcal{E'})$ represents the $(m',n',\xcal',\delta_s)$ BNSI problem, where $m'=|\mathcal{U'}|$, $n'=|\mathcal{P'}|$ and $\xcal'$ is the tuple $(\xcal'_i=\xcal_i \cap \rho,~\forall u_i \in \mathcal{U'})$. In other words, the $(m',n',\xcal',\delta_s)$ BNSI subproblem is derived from the $(m,n,\xcal,\delta_s)$ BNSI problem by deleting some information symbols from the \emph{information symbol set} of the original BNSI problem. 
\begin{lemma} \label{subprob} 
Let $N_{q,opt}(m',n',\xcal',\delta_s)$ be the optimal codelength over $\Fb_q$ for the $(m',n',\xcal',\delta_s)$ BNSI problem. Then $N_{q,opt}(m',n',\xcal',\delta_s)$ satisfies the following property,
\begin{equation*}
N_{q,opt}(m',n',\xcal',\delta_s) \leq N_{q,opt}(m,n,\xcal,\delta_s).
\end{equation*}
\end{lemma} 
\begin{proof}
 In the $(m',n',\xcal',\delta_s)$ BNSI subproblem, the size of the \emph{demanded information symbol index set} for each user is reduced compared to the original BNSI problem. Consider a valid encoder matrix $\Lb$ with optimal codelength $N_{q,opt}(m,n,\xcal,\delta_s)$ for the $(m,n,\xcal,\delta_s)$ BNSI problem. Now in $(m',n',\xcal',\delta_s)$ BNSI subproblem represented by the subgraph $\mathcal{B'}=(\mathcal{U'},\mathcal{P'},\mathcal{E'})$ of $\mathcal{B}$, any user $u_i \in \mathcal{U'}$ has \mbox{$\xcal'_i=\xcal_i \cap \rho \subset \xcal_i$} and $\ycal'_i =\ycal_i \cap \rho \subset \ycal_i$. As  \textit{rowspan}$\{\Lb_{\ycal'_i}\} \subset$ \textit{rowspan}$\{\Lb_{\ycal_i}\}$ and $\Lb_{\xcal'_i}$ is a submatrix of $\Lb_{\xcal'_i}$, in submatrix $\Lb_{\rho}$ any non-zero linear combination of $2\delta_s$ or fewer rows of $\Lb_{\xcal'_i}$ is not in \textit{rowspan}$\{\Lb_{\ycal'_i}\}$. Therefore using Corollary \ref{corr1}, we conclude that $\Lb_{\rho}$ is a valid encoder matrix for $(m',n',\xcal',\delta_s)$ BNSI problem with codelength $N_{q,opt}(m,n,\xcal,\delta_s)$. Thus the optimal codelength $N_{q,opt}(m',n',\xcal',\delta_s)$ for $(m',n',\xcal',\delta_s)$ BNSI problem does not exceed $N_{q,opt}(m,n,\xcal,\delta_s)$. \end{proof}
\vspace{2mm}
\subsubsection{Lower bound based on size of the demanded information symbol index set of each user}    \label{lb_size}
Consider an $(m,n,\xcal,\delta_s)$ BNSI problem represented by the bipartite graph $\mathcal{B}$. Now we obtain the following lower bound.
\begin{theorem} \label{lowerbound1}
Suppose $S=\{i \in [m]~\vert~|\mathcal{X}_i| \in [2\delta_s]\}$ and let $\xcal_S=\bigcup_{i\in S}{\xcal_i}$. Then the optimal codelength $N_{q,opt}$ over $\Fb_q$ for the $(m,n,\xcal,\delta_s)$ BNSI problem satisfies,
\begin{equation*}
N_{q,opt} \geq |\xcal_S|+ \min\{2\delta
_s,n-|\xcal_S|\}.
\end{equation*}
\end{theorem}
\begin{proof}
To derive the lower bound, first we will show that for any subgraph $\mathcal{B'}$ of $\mathcal{B}$ induced by the information symbols indexed by $\xcal_S$ and any $\min\{2\delta_s,n-|\xcal_S|\}$ of the remaining information symbols, the set \mbox{$\Phi(\mathcal{B'})=\phi$}. Then from Theorem \ref{thm2} the optimal codelength $N_{q,opt}(\mathcal{B'})$ over $\Fb_q$ for the subgraph  $\mathcal{B'}$ will be $|\xcal_S|+ \min\{2\delta_s,n-|\xcal_S|\}$, and then using Lemma \ref{subprob}, we have $N_{q,opt}(\mathcal{B}) \geq N_{q,opt}(\mathcal{B'})=|\xcal_S|+ \min\{2\delta_s,n-|\xcal_S|\}$.

 Now to show $\Phi(\mathcal{B'})=\phi$, we will use Algorithm~2. At first Algorithm~2 will take the bipartite graph $\mathcal{B'}$ as input and check whether the size of its \emph{information symbol set} is greater than $2\delta_s$ or not. Now we can have $2$ cases, \textit{Case I.} $S=\phi$ or \textit{Case II.} $S \neq \phi$.
 
\textit{Case I:} If $S=\phi$, $|\xcal_s|=0$. Then the size of the \emph{information symbol set} is \mbox{$\min\{2\delta
_s,n-|\xcal_S|\}$} which is at the most $2\delta_s$. As a result for this case $\Phi(\mathcal{B'})=\phi$ from Algorithm~2.

\textit{Case II:} If $S \neq \phi$, $|\xcal_s|>0$. Hence, the size of the \emph{information symbol set} could be at least $2\delta_s+1$. If so, the bipartite graph $\mathcal{B'}$ will go through the iteration steps in the \emph{while} loop in Algorithm~2. In each step, one \emph{user node} with index from $S$ and its associated \emph{demanded information symbols} will be removed from the bipartite graph $\mathcal{B'}$ since the degree of each of these user nodes is at the most $2\delta_s$. After removing all the information symbols indexed with $\xcal_S$, the remaining number of packets present in the graph will be \mbox{$\min\{2\delta
_s,n-|\xcal_S|\}$} which is at the most $2\delta_s$. Therefore the algorithm will conclude that $\Phi(\mathcal{B'})=\phi$.  
\end{proof}
\vspace{3mm}
\subsubsection{Lower bound based on the set $\Phi(\bcal)$}  \label{lb_phi}
Using Theorem \ref{thm2}, we now provide another lower bound on $N_{q,opt}$ of a BNSI problem. We are interested in a subset $\mathsf{B} \subseteq [n]$ such that the subgraph induced by $x_{\mathsf{B}} \subseteq \mathcal{P}$ denoted by $\mathcal{B}_{x_{\mathsf{B}}}$ satisfies $\Phi(\mathcal{B}_{x_{\mathsf{B}}})=\phi$. Suppose $\mathsf{B}_{max}$ denotes such a $\mathsf{B}$ with largest size. Now the following lower bound holds. 
\begin{theorem} \label{bmax}
The optimal codelength $N_{q,opt}$ over $\Fb_q$ of the $(m,n,\xcal,\delta_s)$ BNSI problem satisfies,
\begin{equation*}
N_{q,opt}(m,n,\xcal,\delta_s) \geq |\mathsf{B}_{max}|.
\end{equation*}
\end{theorem}
\begin{proof}
From Theorem \ref{thm2}, we have that for any choice of $\mathsf{B}$ with $\Phi(\mathsf{B})=\phi$, the optimal codelength for the BNSI problem represented by the subgraph induced by $\mathsf{B}$ is $|\mathsf{B}|$. As $\mathsf{B}_{max}$ denotes such $\mathsf{B}$ with largest size, it holds that $|\mathsf{B}| \leq |\mathsf{B}_{max}|$ for all $\mathsf{B}$ such that $\Phi(\mathcal{B}_{x_{\mathsf{B}}})=\phi$. Now using Lemma \ref{subprob}, $N_{q,opt}(m,n,\xcal,\delta_s) \geq |\mathsf{B}_{max}|$. 
\end{proof} 
Now, we will derive a lemma that will provide a comparison between two the lower bounds given in Theorems \ref{lowerbound1} and \ref{bmax}.
\begin{lemma} \label{compare}
Let $\mathsf{B}_{max}$ be a largest subset of $[n]$ such that $\Phi(\mathcal{B}_{x_{\mathsf{B}_{max}}})=\phi$. Also let $\xcal_S=\bigcup_{i\in S}{\xcal_i}$ where, \mbox{$S=\{i \in [m]~\vert~|\mathcal{X}_i| \in [2\delta_s]\}$}. Then, $|\mathsf{B}_{max}| \geq |\xcal_S|+ \min\{2\delta
_s,n-|\xcal_S|\}$.
\end{lemma}
\begin{proof}
In the proof of Theorem \ref{lowerbound1}, we have already shown that for any subgraph $\mathcal{B'}$ of $\mathcal{B}$ induced by the information symbols indexed by $\xcal_S$ and any of the remaining $\min\{2\delta_s,n-|\xcal_S|\}$ information symbols, the set \mbox{$\Phi(\mathcal{B'})=\phi$}. Therefore these information symbols constitute a set $\mathsf{B}$ such that $\Phi(\mathcal{B}_{x_{\mathsf{B}}})=\phi$. Since $\mathsf{B}_{max}$ is a set of largest size among all $\mathsf{B}$ with the property $\Phi(\mathcal{B}_{x_{\mathsf{B}}})=\phi$ the inequality in Lemma \ref{compare} holds. 
\end{proof}
From Lemma \ref{compare}, we can remark that given an $(m,n,\xcal,\delta_s)$ BNSI problem the lower bound on optimal codelength $N_{q,opt}$ found in Theorem \ref{bmax} is at least as good as the lower bound found in Theorem \ref{lowerbound1}. However the lower bound in Theorem \ref{lowerbound1} can be calculated easily while we do not know of an efficient technique to compute $|\mathsf{B}_{max}|$.     

\begin{example} \label{lbexp}
Consider the BNSI problem scenario mentioned in Example \ref{exmp1}. For this problem scenario $|\xcal_S|=0$, hence from Theorem \ref{lowerbound1} we have $N_{q,opt} \geq \min\{2\delta_s,n\}=2$. Also we can check that any subset of $\{1,2,3,4\}$ of size $3$, i.e., $\{1,2,3\}$, $\{1,3,4\}$, $\{2,3,4\}$, serves as $\mathsf{B}_{max}$. So, from Theorem \ref{bmax} $N_{q,opt} \geq |\mathsf{B}_{max}|=3$.  A valid encoding and decoding scheme over $\Fb_2$ is given in Example \ref{exmp3} that meets this lower bound for this scenario. Further, this scheme can be easily generalized to any finite field $\Fb_q$. Hence, $N_{q,opt}=3$ for this problem for any $\Fb_q$.
\end{example}

\subsection{Construction of encoder matrix $\Lb$ based on linear error correcting codes} \label{linecc}
In this subsection, we describe a construction of a valid encoder matrix $\Lb$ for an $(m,n,\xcal,\delta_s)$ BNSI problem based on linear error correcting codes over $\Fb_q$. Consider a parity check matrix \mbox{$\Hb \in \Fb_q^{(n'-k') \times n'}$} of an $[n',k']$ linear error correcting code over $\Fb_q$ where $n'$, $k'$ denote the blocklength and the dimension of the code,  respectively. Let $d_{min}$ be the \emph{minimum distance} of the code. Then any set of ($d_{min}-1$) columns of $\Hb$ are linearly independent and at least one set of $d_{min}$ columns are linearly dependent \cite{huffman_2010}. Define, $\eta =2\delta_s+ \max_{i \in [m]}{|\ycal_i|}$, where $\ycal_i$ is the index set of the messages that are not demanded by $i^{th}$ user in the $(m,n,\xcal,\delta_s)$ BNSI problem. Now if $d_{min} \geq \eta+1$, $n'=n$ and $\Lb=\Hb^T$, the following lemma holds.
\begin{lemma}
If $\Hb$ is a parity check matrix of an $[n,k',d_{min}]$ code over $\Fb_q$ with $d_{min} \geq 2\delta_s+ \max_{i \in [m]}{|\ycal_i|}+1$, then $\Lb=\Hb^T$ is a valid encoder matrix for the $(m,n,\xcal,\delta_s)$ BNSI problem.
\end{lemma}
\begin{proof}
From Corollary \ref{corr1} we know that to be a valid encoder matrix for an $(m,n,\xcal,\delta_s)$ BNSI problem it is sufficient that any $|\ycal_i|+2\delta_s$ rows in $\Lb$ are linearly independent for each $i \in [m]$. As $(|\ycal_i|+2\delta_s) \leq \eta$, if we consider $\Lb=\Hb^T$ and $d_{min} \geq \eta+1$, $\Lb$ has any set of $\eta$ rows as linearly independent. In particular for any $i^{th}$ user, $i \in [m]$, any $2\delta_s$ or fewer rows of $\Lb_{\xcal_i}$ and all the rows of $\Lb_{\ycal_i}$ together form a linearly independent  set. Therefore, $\Lb$ is a valid encoder matrix for the $(m,n,\xcal,\delta_s)$ BNSI problem.
\end{proof}
 We can utilize a linear error correcting code having blocklength $n$ and $d_{min} \geq \eta+1$ over $\Fb_q$ with maximum possible dimension $k'$ such that $n-k'$ is minimized. Then $\Lb$ will be the transpose of a parity check matrix of the error correcting code with codelength $N=(n-k')$. 
 
\begin{example}  
 Suppose $m=4$, $n=6$, $\xcal_1=\{1,2,3,4\}$, $\xcal_2=\{2,3,4,5\}$, $\xcal_3=\{1,3,4,5,6\}$, $\xcal_4=\{2,3,4,5,6\}$ and $\delta_s=1$. Therefore $\eta=2\delta_s+ \max_{i \in [m]}{|\ycal_i|}=2+2=4$. We now use a $[6,k']$ linear error correcting code over $\Fb_q$ with maximum possible $k'$ having $d_{min} \geq 5$. From~\cite{codetables}, we can find that such codes over $\Fb_2$ are $[6,1,6]$ and $[6,1,5]$ and the resulting codelength $N$ for both the cases will be $5$. Over $\Fb_5$ such a linear error correcting code is $[6,2,5]$ and the resulting codelength $N=4$.  
\end{example}

Among all the linear error correcting codes over $\Fb_q$ having blocklength $n$ and $d_{min}= \eta+1$, the dimension $k'$ will be maximum for \emph{Maximum Distance Separable} (MDS) codes if such an MDS code exists over $\Fb_q$. Suppose $\Lb$ is a valid encoder matrix for an $(m,n,\xcal,\delta_s)$ BNSI problem constructed based on the transpose of a parity check matrix $\Hb$ of an MDS code over $\Fb_q$($q \geq n$) having blocklength $n'= n$ and $d_{min}= \eta+1$. Then the dimension of the code $k'= n'-d_{min}+1 = (n-\eta)^+ = (n-2\delta_s - \max_{i \in [m]}{|\ycal_i|})^+=(n-2\delta_s - \max_{i \in [m]}{(n-|\xcal_i|)})^+=\min_{i \in [m]}{(|\mathcal{X}_i|-2\delta_s)^+}$, where $x^+=x~\text{for}~x \geq 0~\text{and}~x^+=0~\text{for}~x<0$.

\begin{example} \label{exmp4} 
 Consider the BNSI problem scenario, where $m=4$, $n=10$, $\delta_s=1$, $\xcal_1=\{1,3,5,7,9\}$, $\xcal_2=\{2,4,6,8,10\}$, $\xcal_3=\{1,2,4,6,8,10\}$, $\xcal_4=\{3,4,5,6,7,9\}$. For this example, $k'=\min_{i \in [m]}{(|\mathcal{X}_i|-2\delta_s)^+}$ $=3$. Over $\Fb_{16}$ there exists an $[10,3]$ linear error correcting code with $d_{min}=8$  which is an MDS code. 
The transpose of the parity-check matrix of this code is a valid encoder matrix for the BNSI problem. Note that using this MDS code we save $3$ transmissions compared to uncoded scheme.   
\end{example}
\subsection{Upper Bounds on $N_{q,opt}$ }  \label{ub}
Here, we will describe three upper bounds on $N_{q,opt}$ of an $(m,n,\xcal,\delta_s)$ BNSI problem. First one is based on the code construction from linear error correcting codes as given in Section \ref{linecc}, the second one is based on \emph{disjoint elements} of the set $\Phi(\bcal)$ defined over the bipartite graph $\bcal$ which represents the $(m,n,\xcal,\delta_s)$ BNSI problem and the last one is based on partitioning the set of information symbols. 
\vspace{2mm}
\subsubsection{Upper bound based on linear error correcting codes}  \label{ub_ecc}
From Section \ref{linecc}, we have a valid encoder matrix of an $(m,n,\xcal,\delta_s)$ BNSI problem with codelength $N=n'-k'$ derived from an $[n',k']$ linear error correcting code having blocklength $n'=n$ and $d_{min} \geq \eta+1$ with maximum possible dimension $k'$. Let $k(q,n,d_{min})$ be the largest possible dimension among all linear error correcting codes over $\Fb_q$ with blocklength $n$ and minimum distance at least $d_{min}$. Then we have $N_{q,opt} \leq n-k(q,n,d_{min})$. From this inequality condition, we now obtain an upper bound on the optimal codelength $N_{q,opt}$.
\begin{theorem} \label{thm12}
The optimal codelength $N_{q,opt}$ over $\Fb_q~(q \geq n)$ for an $(m,n,\xcal,\delta_s)$ BNSI problem satisfies
\begin{equation*} 
N_{q,opt} \leq n-\min\limits_{i \in [m]}{(|\mathcal{X}_i|-2\delta_s)^+}.  
\end{equation*}
\end{theorem}
\begin{proof}
The codelength of a valid coding scheme for an $(m,n,\xcal,\delta_s)$ BNSI problem based on linear error correcting codes as given in Section \ref{linecc} will be minimum if the encoder matrix $\Lb$ is derived from an $[n',k']$ linear MDS code with blocklength $n'=n$, dimension $k'$ and $d_{min} = \eta+1$ if such an MDS code exists. We have the dimension of such MDS code is $k'=\min_{i \in [m]}{(|\mathcal{X}_i|-2\delta_s)^+}$. If $q \geq n$ then such an MDS code exists over $\Fb_q$. Hence, the upper bound in Theorem \ref{thm12} holds.
\end{proof}

\subsubsection{Upper bound based on disjoint elements of $\Phi(\mathcal{B})$}   \label{ub_min}
Now we provide an upper bound on optimal codelength $N_{q,opt}$ for an $(m,n,\xcal,\delta_s)$ BNSI problem represented by the bipartite graph $\mathcal{B}=(\mathcal{U},\mathcal{P},\mathcal{E})$. This upper bound is motivated by \emph{Cycle-Covering scheme} for Index Coding~\cite{MAZ_IEEE_IT_13,CASL_ISIT_11}. For each element $\cs \in \Phi(\mathcal{B})$, the subgraph induced by $x_{\cs}$ denoted as $\mathcal{B}_{\cs}=(\mathcal{U}_{\cs},x_{\cs},\mathcal{E}_{\cs})$ represents the $(m_{\cs},n_{\cs},\xcal_{\cs},\delta_s)$ BNSI problem, where $m_{\cs}=|\mathcal{U}_{\cs}|=|\{u_i \in \mathcal{U}~|\cs \cap \xcal_i \neq \phi\}|$, $n_{\cs}=|\cs|$, $\xcal_{\cs}=\{\xcal_{\cs,i}~|\xcal_{\cs,i}=\xcal_i \cap \cs, \forall u_i \in \mathcal{U}_{\cs}\}$ and $\mathcal{E}_{\cs}=\{\{u_i,x_j\} \in \mathcal{E}~|u_i \in \mathcal{U}_{\cs},~j \in \cs\}$. Now since $\cs \in \Phi(\bcal)$ it can be noticed that $\forall u_i \in \mathcal{U}_{\cs}$, the degree of $u
_i$ in $\bcal_{\cs}$, $deg(u_i) \geq 2\delta_s+1$. Therefore we can use the simple coding scheme described in Section~\ref{sec:simple_coding} on $(m_{\cs},n_{\cs},\xcal_{\cs},\delta_s)$ BNSI problem to save one transmission compared to uncoded transmission. Therefore the length of this code to transmit all the information symbols indexed by $\cs \subseteq [n]$ over $\Fb_q$ is $N_{\cs}=|\cs|-1$. For some integer $K$, let $\cs_1,\cs_2,\dots,\cs_K \in \Phi(\bcal)$ and $R=[n] \setminus (\cs_1 \cup \cs_2 \cup \dots \cup \cs_K)$. Given such a collection of elements of $\Phi(\bcal)$, we design a valid coding scheme as follows. We apply the coding scheme described in Section \ref{sec:simple_coding} on each element $\cs_1,\cs_2,\dots,\cs_K$ and transmit the information symbols indexed by the set $R$ uncoded. The codelength for this scheme is
\begin{align*}
N &= \sum_{i=1}^{K}{(|\cs_i|-1)}+|R| 
= \sum_{i=1}^{K}{|\cs_i|}-K+|R|.
\end{align*}
\begin{lemma} \label{disjoint}
Let $N$ be the codelength of the linear coding scheme based on the set $\cs_1,\cs_2,\dots,\cs_K \in \Phi(\bcal)$. Then there exist disjoint $\cs'_1,\cs'_2,\dots,\cs'_{K'} \in \Phi(\bcal)$ such that $K' \leq K$ and the codelength $N'$ of the linear coding scheme based on $\cs'_1,\cs'_2,\dots,\cs'_{K'}$ is at the most $N$. 
\end{lemma}
\begin{proof}
From the set $\cs_1,\cs_2,\dots,\cs_K \in \Phi(\bcal)$, we construct $K$ sets $\cs'_1,\cs'_2,\dots,\cs'_K$ as follows, $\cs'_1=\cs_1$, $\cs'_2=\cs_2 \setminus \cs_1$, $\cs'_3=\cs_3 \setminus (\cs_1 \cup \cs_2)$, $\dots$, $\cs'_K=\cs_K \setminus (\cs_1 \cup \cs_2 \cup \dots \cup \cs_{K-1})$. Note that $\cs'_1,\cs'_2,\dots, \cs'_K$ are disjoint and $|\cs'_1| \leq |\cs_1|$, $|\cs'_2| \leq |\cs_2|$, $\dots$, $|\cs'_K| \leq |\cs_K|$. Now we categorize $\cs'_1,\cs'_2,\dots, \cs'_K$ into two sets $\cs'$ and $R'$ as follows, if $\cs'_i \in \Phi(\bcal)$ where $i \in [K]$, we keep the set $\cs'_i$ in the set $\cs'$ otherwise keep the set $\cs'_i$ in the set $R'$. Without loss of generality we assume that the first $K'$ sets, $K' \leq K$, among $\cs'_1,\cs'_2,\dots, \cs'_K$ belongs to $\cs'$. Then $\cs'=\{\cs'_1, \cs'_2, \dots, \cs'_{K'}\}$ and $R'=\cs'_{K'+1} \cup \cs'_{K'+2} \cup \dots \cup \cs'_{K}$. Let $R_{mod}=R \cup R'$. Note that $R$ and $R'$ are disjoint. Now we design a valid coding scheme as follows, we apply the coding scheme described in Section \ref{sec:simple_coding} on each element of $\cs'$ and send the information symbols indexed by the set $R_{mod}$ uncoded. Therefore the codelength for this scheme is 
\begin{align*}
N' &= \sum_{i=1}^{K'}{|\cs'_i|-1} + |R_{mod}|\\
&= \sum_{i=1}^{K'}{|\cs'_i|-1} + |R| + |R'|\\
&= \sum_{i=1}^{K'}{|\cs'_i|-1} + |R| + \sum_{i=K'+1}^{K}{|\cs'_i|} \\
&= \sum_{i=1}^{K}{|\cs'_i|}-K' + |R|.
\end{align*} 
For any $K' < i \leq K$, 
$|\cs_i|-|\cs'_i| \geq 1$. Hence we have
\begin{align*}
\sum_{i=K'+1}^{K}{|\cs_i|-|\cs'_i|} &\geq (K-K'), \text{ and thus}\\
\sum_{i=1}^{K}{|\cs_i|-|\cs'_i|} &\geq (K-K')\\
\sum_{i=1}^{K}{|\cs'_i|}-K' &\leq \sum_{i=1}^{K}{|\cs_i|}-K.
\end{align*}
Therefore $N'= \sum_{i=1}^{K}{|\cs'_i|}-K' + |R| \leq \sum_{i=1}^{K}{|\cs_i|}-K +|R|=N$. Hence the lemma holds.
\end{proof}  
Now applying our designed coding scheme on disjoint elements of $\Phi(\bcal)$ we have the following upper bound on the optimal codelength $N_{q,opt}$.
\begin{theorem} \label{thm13}
Let $\mathfrak{C}$ be a largest collection of disjoint elements of $\Phi(\bcal)$ for an $(m,n,\xcal,\delta_s)$ BNSI problem and $\mathfrak{R}=[n] \setminus \mathfrak{C}$. The optimal codelength $N_{q,opt}$ for the $(m,n,\xcal,\delta_s)$ BNSI problem over $\Fb_q$ satisfies
\begin{equation*} 
N_{q,opt} \leq n-|\mathfrak{C}|.  
\end{equation*}
\end{theorem} 
\begin{proof}
Applying the coding scheme mentioned in Section \ref{sec:simple_coding} on each element of $\mathfrak{C}$, we can save one transmission compared to uncoded scheme. Thereby we can save $|\mathfrak{C}|$ transmission for the collection $\mathfrak{C}$.  
\end{proof}
\begin{lemma}   \label{equality}
Let $\mathfrak{C}=\{\cs_1,\cs_2,\dots,\cs_K\}$ and $i_1 \in \cs_1, i_2 \in \cs_2,\dots, i_k \in \cs_K$ are such that the subgraph $\bcal'$ of the bipartite graph $\bcal$ induced by $\pcal'=\pcal \setminus \{x_{i_1},x_{i_2},\dots,x_{i_K}\}$ satisfies $\Phi(\bcal')=\phi$. Then the optimal codelength $N_{q,opt}$ over $\Fb_q$ satisfies
\begin{equation*} 
N_{q,opt} = n-|\mathfrak{C}|.  
\end{equation*} 
\end{lemma}
\begin{proof}
From Theorem \ref{thm13} we have the upper bound on $N_{q,opt}$. It remains to show that $N_{q,opt} \geq n-|\mathfrak{C}|$. The number of information symbols in $\bcal'$ is $n-|\mathfrak{C}|$. As $\Phi(\bcal')=\phi$,  from Theorem~\ref{thm2} we have $N_{q,opt}(\bcal')=n-|\mathfrak{C}|$. Now using Lemma \ref{subprob} we have $N_{q,opt}(\bcal) \geq N_{q,opt}(\bcal')=n-|\mathfrak{C}|$.
\end{proof}

If we apply the coding scheme derived from a linear MDS code on each element of $\mathfrak{C}$ and transmit the information symbols index by $\mathfrak{R}$ uncoded then we can have the following upper bound on $N_{q,opt}$.
\begin{theorem} \label{thm14}
Suppose the subgraph of $\mathcal{B}$ induced by the information symbols indexed by a set $\cs \in \mathfrak{C}$ is denoted by \mbox{$\mathcal{B}_{\cs}=(\mathcal{U}_{\cs},x_{\cs},\mathcal{E}_{\cs})$} and define \mbox{$d_{\cs}= \min_{u_i \in \mathcal{U}_{\cs}}{(|\xcal_i \cap {\cs}|-2\delta_s)^+}$}. Then the optimal length $N_{q,opt}$ over $\Fb_q$ where $q \geq \max_{\cs \in \mathfrak{C}}{|\cs|}$ for the $(m,n,\xcal,\delta_s)$ BNSI problem represented by the bipartite graph $\mathcal{B}$ satisfies
\begin{equation*} 
N_{q,opt} \leq n- \sum\limits_{\cs \in \mathfrak{C}}{d_{\cs}}.  
\end{equation*}
\end{theorem}
\begin{proof}
To transmit the information symbols indexed by the set $\cs \in \mathfrak{C}$, if we use an encoder matrix derived from an $[n',k']$ linear MDS code with blocklength $n'=|\cs|$ and dimension $k'=d_{\cs}$ then from Theorem \ref{thm12} it is known that we can save $d_{\cs}$ transmissions compared to uncoded scheme. Such an MDS code exists over $\Fb_q$ if $q \geq \max_{\cs \in \mathfrak{C}}{|\cs|}$. As the elements in $\mathfrak{C}$ are disjoint then the total number of transmissions that can be saved is $\sum_{\cs \in \mathfrak{C}}{d_{\cs}}$. Therefore the upper bound in Theorem \ref{thm14} holds.  
\end{proof} 

\begin{remark}
For a given $(m,n,\xcal,\delta_s)$ BNSI problem, in general the upper bound on optimal codelength $N_{q,opt}$ found in Theorem \ref{thm14} is at least as good as the upper bound found in Theorem \ref{thm13} because in each \mbox{$\cs \in \mathfrak{C}$}, $d_{\cs} \geq 1$, therefore $\sum_{\cs \in \mathfrak{C}}{d_{\cs}} \geq |\mathfrak{C}|$. In other words, if we apply coding scheme mentioned in Section~\ref{sec:simple_coding} on each $\cs \in \mathfrak{C}$ we can save exactly one transmission compared to uncoded scheme whereas if we apply coding scheme based on an linear MDS code we can save \emph{at least} one transmission. However for the upper bound given in Theorem \ref{thm14} we need the finite field size $q$ to be large while Theorem~\ref{thm13} holds for any $q$. 
\end{remark}

\begin{example} 
Consider a BNSI problem scenario where $m=3$, $n=10$, $\delta_s=1$, $\xcal_1=\{1,2,3,9\}$, $\xcal_2=\{4,5,6,10\}$, $\xcal_3=\{7,8\}$. We can find that a possible choice of $\mathfrak{C}=\{\{1,2,3,9\},\{4,5,6,10\}\}$. If the finite field size $q=2$, then using Theorem~\ref{thm12} we obtain $N_{q,opt} \leq n=10$ whereas using Theorem~\ref{thm13} we obtain \mbox{$N_{q,opt} \leq n-2=8$}. However if $q \geq 4$ then using Theorem~\ref{thm14} we obtain $N_{q,opt} \leq n-4=6$. 
\end{example}
\begin{example} 
Consider a BNSI problem scenario where $m=4$, $n=7$, $\delta_s=1$, $\xcal_1=\{1,3,5\}$, $\xcal_2=\{2,4,6\}$, $\xcal_3=\{3,6,7\}$, $\xcal_4=\{4,5,6\}$. We can find that a possible choice of $\mathfrak{C}=\{\{1,3,5\},\{2,4,6\}\}$. For the finite field size $q=2$, using Theorem~\ref{thm13} we obtain \mbox{$N_{q,opt} \leq n-2=5$}. Now we are deleting one index from each element of $\mathfrak{C}$. Suppose $\bcal'$ is the subgraph of $\bcal$ induced by the information symbols indexed by the remaining indices after deleting any one index from each element in $\mathfrak{C}$. We can check that $\Phi(\bcal') = \phi$. Hence applying Lemma \ref{equality} we have $N_{q,opt}=5$. 
\end{example}

\subsubsection{Upper bound based on partitioning the maximum element of $\Phi(\mathcal{B})$}   \label{ub_par}
We now provide another upper bound on the optimal codelenth $N_{q,opt}$ for the $(m,n,\xcal,\delta_s)$ BNSI problem represented by the bipartite graph $\bcal=(\ucal,\pcal,\ecal)$ based on the partitioning the maximum element of $\Phi(\bcal)$. This upper bound is motivated the by the \emph{partition multicast} scheme for Index Coding as described in \cite{tehrani2012bipartite,iscod_1998}.  We will now show that the set $\cs$ output by Algorithm~2 is a maximal element of $\Phi(\bcal)$ and then show that $\cs$ is the unique maximal element in $\Phi(\bcal)$. Hence $\cs$ is the maximum element in $\Phi(\bcal)$.
\begin{lemma}  \label{maximal_element}
If $\Phi(\mathcal{B}) \neq \phi$, the index set of the information symbols $\mathsf{C}$ output by Algorithm~2 is a maximal element of $\Phi(\mathcal{B})$.
\end{lemma}  
\begin{proof}
To show that set $\mathsf{C}$ is a maximal element, we will show that if we further add any set of information symbols with the set $\mathsf{C}$ then the resulting set will not be an element of $\Phi(\mathcal{B})$. Algorithm~2 keeps deleting user $u_i$ and its corresponding $\xcal_i$ iteratively until a $\mathsf{C} \in \Phi(\bcal)$ is found. Suppose in Algorithm~2 after deleting $t$ users from \textit{user-set} $\mathcal{U}$ we found the set $\mathsf{C}$ and $\mathcal{U}_{del}=\{u_1,u_2,\dots,u_t\}$ denotes the set of deleted users, where without loss of generality we have assumed that $u_1$ is first deleted user and then $u_2,u_3,\dots,u_t$ are deleted consecutively. The set of deleted information symbols denoted by $\xcal_{del}= \xcal_1 \cup \xcal_2 \cup \dots \cup \xcal_t$. Suppose we are adding a set of information symbols indexed by $\xcal_A, \xcal_A \subseteq \xcal_{del}$ with the set $\mathsf{C}$. Let $i$ be the smallest integer such that $\xcal_i \cap \xcal_A \neq \phi$.  Now in the subgraph $\mathcal{B}_{x_{(\mathsf{C}~\cup~\xcal_A)}}$ induced by the information symbols indexed by the set ${\mathsf{C} \cup \xcal_A}$, \textit{deg}($u_i$) = $|\xcal_i \cap (\mathsf{C} \cup \xcal_A)| = |\xcal_i \cap (\xcal_A \cup \cs)| \leq |x_{\xcal_i} \cap (x_{\xcal_i} \cup x_{\xcal_{i+1}} \cup \dots \cup x_{\xcal_t} \cup \mathsf{C})|$ = \textit{deg}($u_i$) in $\mathcal{B}_{x_{(\xcal_i \cup \xcal_{i+1} \cup \dots \cup \xcal_t) \cup \mathsf{C}}} \in [2\delta_s]$. This is due to the fact that Algorithm~2 deletes $u_i \in \ucal_{del}$ from the bipartite graph $\mathcal{B}_{x_{(\xcal_i \cup \xcal_{i+1} \cup \dots \cup \xcal_t) \cup \mathsf{C}}}$ as its \textit{degree} is at the most 2$\delta_s$. 
So, the index set $(\mathsf{C} \cup \xcal_A)$ is not an element of $\Phi(\mathcal{B})$ which shows that the set $\mathsf{C}$ is a maximal element of $\Phi(\mathcal{B})$.
\end{proof}
\begin{lemma}  \label{unique}
$\Phi(\mathcal{B})$ contains a unique maximal element.
\end{lemma}
\begin{proof}
We will use proof by contradiction. Suppose $\mathsf{C}$ and $\mathsf{C'}$ are two maximal elements of $\Phi(\mathcal{B})$ such that $\mathsf{C} \neq \mathsf{C'}$. Recall that for any $i \in [m]$, $|\xcal_i \cap \mathsf{C}| \notin [2\delta_s]$ and $|\xcal_i \cap \mathsf{C'}| \notin [2\delta_s]$.  Consider the set $\mathsf{C} \cup \mathsf{C'}$ which is a subset of $[n]$. Now for any $i^{th}$ user, $i \in [m]$, $\xcal_i$ will satisfy one of the four following possibilities, (i) $\xcal_i \cap \mathsf{C}= \phi$ and $\xcal_i \cap \mathsf{C'}= \phi$, (ii) $\xcal_i \cap \mathsf{C}= \phi$ and $\xcal_i \cap \mathsf{C'} \neq \phi$, (iii) $\xcal_i \cap \mathsf{C} \neq \phi$ and $\xcal_i \cap \mathsf{C'}= \phi$, (iv) $\xcal_i \cap \mathsf{C} \neq \phi$ and $\xcal_i \cap \mathsf{C'} \neq \phi$. From the knowledge that $|\xcal_i \cap \mathsf{C}|$ and $|\xcal_i \cap \mathsf{C'}|$ is either $0$ or at least $2\delta_s+1$, we can conclude that $|\xcal_i \cap (\mathsf{C} \cup \mathsf{C'})|$ is either $0$ or at least $2\delta_s+1$. Therefore $\mathsf{C} \cup \mathsf{C'}$ is an element of $\Phi(\mathcal{B})$ and $|(\mathsf{C} \cup \mathsf{C'})| > |\mathsf{C}|,|\mathsf{C'}|$ which contradicts the maximality of both  $\mathsf{C}$ and $\mathsf{C'}$. Hence the lemma holds.
\end{proof}
From now onward we denote the maximum or the unique maximal element of $\Phi(\mathcal{B})$ as  $\mathsf{C}_{max}$. We now provide a result that will provide some knowledge regarding to those information symbols that do not belong to the set $\mathsf{C}_{max}$.
\begin{lemma}   \label{acyclic}
Suppose an $(m,n,\xcal,\delta_s)$ BNSI problem is represented by the bipartite graph $\mathcal{B}=(\mathcal{U},\mathcal{P},\mathcal{E})$ and $\mathsf{C}_{max}$ denotes the maximum element of $\Phi(\mathcal{B})$. The subgraph $\mathcal{B}'$ of $\mathcal{B}$ induced by the set $\mathcal{P} \setminus x_{\mathsf{C}_{max}}$ satisfies $\Phi(\mathcal{B'})=\phi$. 
\end{lemma} 
\begin{proof}
We will use proof by contradiction. Suppose $\Phi(\bcal') \neq \phi$ and a set $\mathsf{C'} \in \Phi(\mathcal{B'})$. Consider the set $\mathsf{C}_{max} \cup \mathsf{C'} \subseteq [n]$. Using the same argument used to prove Lemma \ref{unique}, we can conclude that the set $\mathsf{C}_{max} \cup \mathsf{C'}$ is an element of $\Phi(\mathcal{B})$ which contradicts the maximality of $\mathsf{C}_{max}$. Hence, $\Phi(\mathcal{B'})=\phi$.  
\end{proof} 
From Theorem \ref{thm12} we deduce that for an $(m,n,\xcal,\delta_s)$ BNSI problem if we denote $d=\min_{i \in [m]}$ ${(|\xcal_i|-2\delta_s)^+}$, we can save $d$ transmissions compared to uncoded transmission by using an encoder matrix derived from an MDS code over $\Fb_q$ ($q \geq n$) with blocklength $n$ and dimension $d$. We now partition the maximum element $\mathsf{C}_{max}$ of $\Phi(\mathcal{B})$ into $K$ disjoint subsets $\mathsf{S}_1, \mathsf{S}_2, \dots, \mathsf{S}_K \subset \mathsf{C}_{max}$, i.e., for any $a,a' \in [K]$, $a \neq a'$, $\mathsf{S}_{a} \cap \mathsf{S}_{a'}=\phi$ and $\bigcup_{a=1}^{K}{\mathsf{S}_{a}}=\mathsf{C}_{max}$. Note that for each $a \in [K]$, the subgraph $\mathcal{B}_{a}=(\mathcal{U}_{a},x_{\mathsf{S}_{a}},\mathcal{E}_{a})$ induced by $x_{\mathsf{S}_{a}}$ denotes the $(m_{a},n_{a},\xcal_{a},\delta_s)$ BNSI problem where $m_{a}=|\mathcal{U}_{a}|=|\{u_i \in \mathcal{U}~|\mathsf{S}_{a} \cap \xcal_i \neq \phi\}|$, $n_{a}=|\mathsf{S}_{a}|$, \mbox{$\xcal_{a}=\{\xcal_{a,i}~|\xcal_{a,i}=\xcal_i \cap \mathsf{S}_{a}, \forall u_i \in \mathcal{U}_{a}\}$} and $\mathcal{E}_{a}=\{\{u_i,x_j\} \in \mathcal{E}~|u_i \in \mathcal{U}_{a},~j \in \mathsf{S}_{a}\}$. Let $d_{a}=\min_{u_{i'} \in \mathcal{U}_{a}}{(|\xcal_{i'} \cap \mathsf{S}_{a}|-2\delta_s)^+}$. While transmitting the information symbols indexed by   $\mathsf{S}_{a}$, we can save $d_{a}$ transmissions compared to the uncoded scheme by using an encoder matrix derived from an MDS code over $\Fb_q$ ($q \geq |\mathsf{S}_{a}|$) with blocklength $|\mathsf{S}_{a}|$ and dimension $d_{a}$. We encode the symbols in each $\mathsf{S}_{a},~a \in [K]$ independently using this coding scheme. The symbols whose indices are not in $\mathsf{C}_{max}$ are transmitted uncoded. Therefore the total number of transmissions we can save through partitioning is $d_{sum}=\sum_{a=1}^{K}{d_{a}}$. To save maximum transmissions we need to partition the set $\mathsf{C}_{max}$ in such a way that maximizes $d_{sum}$. Therefore the \emph{optimal partitioning} is the solution of the following optimization problem.
\begin{opt}   \label{opt1}
\begin{align*}
\vspace{-8mm}
& \text{maximize}~d_{sum}=\sum_{a=1}^{K}{d_{a}},~~\text{where}~d_{a}=\min_{u_{i'} \in \mathcal{U}_{a}}{(|\xcal_{i'} \cap \mathsf{S}_{a}|-2\delta_s)^+}\\
& \text{subject to}~ 1 \leq K \leq n\\
&\mathsf{S}_1, \mathsf{S}_2, \dots, \mathsf{S}_K \subset \mathsf{C}_{max}~\text{such that}\\
& \text{for any}~a,a' \in [K],a \neq a',~~\mathsf{S}_{a} \cap \mathsf{S}_{a'}=\phi\\ 
& \text{and}~ \bigcup_{a=1}^{K}{\mathsf{S}_{a}}=\mathsf{C}_{max}.
\end{align*}
\end{opt}   

\begin{remark} \label{subopt}
For any $(m,n,\xcal,\delta_s)$ BNSI problem represented by the bipartite graph $\mathcal{B}=(\mathcal{U},\mathcal{P},\mathcal{E})$ if $\mathsf{C}_{max}$ is the only element in $\Phi(\mathcal{B})$ or in other words $|\Phi(\mathcal{B})|=1$, then partitioning $\mathsf{C}_{max}$ into two or more subsets is not optimal. It is trivial to check that if we partition $\mathsf{C}_{max}$, none of the partitions will be an element of $\Phi(\mathcal{B})$. Therefore we can not save any transmission from any of the partition whereas using the full set $\mathsf{C}_{max}$ we can save at  least one transmission.  
\end{remark}
The following upper bound on the optimal codelength $N_{q,opt}$ is a direct result of the optimal partitioning of $\mathsf{C}_{max}$.
\begin{theorem}    \label{thm:ub_par}
Let $D_{sum}$ be the solution to the optimization problem \emph{\textbf{Optimization~1}}. Then optimal codelength $N_{q,opt}$ over $\Fb_q$ satisfies
\begin{equation*} 
N_{q,opt} \leq n-D_{sum}.  
\end{equation*}
\end{theorem} 

\section{BNSI problem and Index Coding} \label{BNSI_IC}
In this section we will show that every BNSI problem is equivalent to an \emph{Index Coding} problem \cite{YBJK_IEEE_IT_11} and using the equivalent Index Coding problem we find a valid encoder matrix for $(m,n,\xcal,\delta_s)$ BNSI problem. We also obtain a lower bound on $N_{q,opt}$ for a $(m,n,\xcal,\delta_s)$ BNSI problem based on this property. 

\subsection{Index Coding with side information}
\emph{Index Coding} \cite{YBJK_IEEE_IT_11} deals with the problem of code design for the transmission of a vector of $n_{IC}$ information symbols or messages denoted as $\xb'_{IC}=(x'_1,x'_2,\dots, x'_{n_{IC}}) \in \Fb_q^{n_{IC}}$ to $m_{IC}$ users denoted as $u'_1,u'_2,\dots,u'_{m_{IC}}$ over a noiseless broadcast channel. It is assumed that $i^{th}$ user $u'_i,~\forall i \in [m_{IC}]$ already knows a part of the transmitted message vector as \emph{side information} denoted as $\xb'_{\xcal_{i,IC}},~\xcal_{i,IC} \subseteq [n_{IC}]$ and demands message $x'_{f(i)},~f(i) \in [n_{IC}]$ where \mbox{$f:[m_{IC}] \rightarrow [n_{IC}],$} such that $f(i) \notin \xcal_{i,IC}$. The set $\xcal_{i,IC}$ is \emph{side information index set} and $f(i)$ is \emph{demanded message index}. Upon denoting $\xcal_{IC}=(\xcal_{1,IC}, \xcal_{2,IC}, \dots, \xcal_{{n_{IC}},IC})$, we describe this Index Coding problem as $(m_{IC},n_{IC},\xcal_{IC},f)$ Index Coding problem. As described in \cite{YBJK_IEEE_IT_11}, a valid \emph{encoding function} over $\Fb_q$ for an $(m_{IC},n_{IC},\xcal_{IC},f)$ Index Coding problem is defined by, 
\begin{equation*}
\mathfrak{E}_{IC}:\Fb_q^{n_{IC}}\rightarrow \Fb_q^{N_{IC}}
\end{equation*}
such that for each user $u'_i,$ \mbox{$i \in [m_{IC}]$} there exists a \emph{decoding function}
\mbox{$\mathfrak{D}_{i,IC}:\Fb_q^{N_{IC}} \times \Fb_q^{|\xcal_{i,IC}|} \rightarrow \Fb_q$}
satisfying the following property:
$\mathfrak{D}_{i,IC}(\mathfrak{E}_{IC}(\xb'_{IC}),\xb'_{\xcal_{i,IC}})=x'_{f(i)}$
for every \mbox{$\xb'_{IC} \in \Fb_q^{n_{IC}}$}.

The design objective is to design a tuple $(\mathfrak{E}_{IC},\mathfrak{D}_{1,IC},\mathfrak{D}_{2,IC},\dots, \mathfrak{D}_{m_{IC},IC})$ of encoding and decoding functions that minimizes the codelength $N_{IC}$ and obtain the \emph{optimal codelength} for the given Index Coding problem which is the minimum codelength among all valid Index Coding schemes. 

A \emph{scalar linear Index Code} for an $(m_{IC},n_{IC},\xcal_{IC},f)$ Index Coding problem is defined as a coding scheme where the encoding function \mbox{$\mathfrak{E}_{IC}:\Fb_q^{n_{IC}}\rightarrow \Fb_q^{N_{IC}}$} is a linear transformation over $\Fb_q$ described as \mbox{$\mathfrak{E}_{IC}(\xb'_{IC})=\xb'_{IC}\Lb_{IC}$}, $\forall \xb'_{IC} \in \Fb_q^{n_{IC}}$, where $\Lb_{IC} \in \Fb_q^{n_{IC} \times N_{IC}}$ is the \emph{encoder matrix for scalar linear Index Code}. The minimum codelength among all valid linear coding schemes for the $(m_{IC},n_{IC},\xcal_{IC},f)$ Index Coding problem over the field $\Fb_q$ will be denoted as $N_{q,opt,IC}(m_{IC},n_{IC},\xcal_{IC},f)$. 

From \cite{Dau_IEEE_J_IT_13} we have a design criterion for a matrix $\Lb_{IC}$ to be a valid encoder matrix for $(m_{IC},n_{IC},\xcal_{IC},f)$ scalar linear Index Coding problem. Following the results in \cite{Dau_IEEE_J_IT_13}, we define the set $\ical_{IC}(q,m_{IC},n_{IC},\xcal_{IC},f)$ or equivalently $\ical_{IC}$ of vectors $\zb$ of length $n$ such that $\zb_{\xcal_{i,IC}}={\bf{0}} \in \Fb_q^{|\xcal_{i,IC}|}$ and $z_{f(i)} \neq 0$ for some choice of $i \in [m_{IC}]$ i.e.,
\begin{equation} \label{i_ic}
\ical(q,m_{IC},n_{IC},\xcal_{IC},f)=\bigcup\limits_{i=1}^{m_{IC}}\{\zb \in \Fb_q^{n_{IC}}~| \zb_{\xcal_{i,IC}}={\bf{0}}~\text{and}~z_{f(i)} \neq 0\}.
\end{equation}
Now from Corollary 3.10 in \cite{Dau_IEEE_J_IT_13} it follows that that $\Lb_{IC}$ is a valid encoder matrix for an $(m_{IC},n_{IC},\xcal_{IC},f)$ scalar linear Index Coding problem if and only if
\begin{equation} \label{L_IC}
\zb\Lb_{IC} \neq {\bf{0}},\quad \forall \zb \in \ical_{IC}.
\end{equation}
It is possible to represent an $(m_{IC},n_{IC},\xcal_{IC},f)$ Index Coding problem by means of a directed bipartite graph as described in \cite{tehrani2012bipartite} which is as follows. The directed bipartite graph \mbox{$\mathcal{B}_{IC}=(\mathcal{U}_{IC},\mathcal{P}_{IC},\mathcal{E}_{IC})$} corresponding to the $(m_{IC},n_{IC},\xcal_{IC},f)$ Index Coding problem consists of the node-sets \mbox{$\mathcal{U}_{IC}=\{u'_1,u'_2,\dots,u'_{m_{IC}}\}$}, $\mathcal{P}_{IC}=\{x'_1,x'_2,\dots,x'_{n_{IC}}\}$, the set of directed edges \mbox{$(u'_i,x'_j) \in \mathcal{E}_{IC}$ if $j \in \xcal_{i,IC}$} and the set of directed edges \mbox{$(x'_j,u'_i) \in \mathcal{E}_{IC}$ if $j=f(i)$}. The set $\mathcal{U}_{IC}$ denotes the \textit{user-set} and $\mathcal{P}_{IC}$ denotes the set of packets or the \textit{information symbol-set}. The directed edges from \textit{user-set} to \textit{information symbol-set} in $\mathcal{E}_{IC}$ denotes the user's side information and directed edges from \textit{information symbol-set} to \textit{user-set} in $\mathcal{E}_{IC}$  denotes the user's demanded message.

\begin{figure}[h!]
\centering
\vspace{-2mm}
\includegraphics[width=2in]{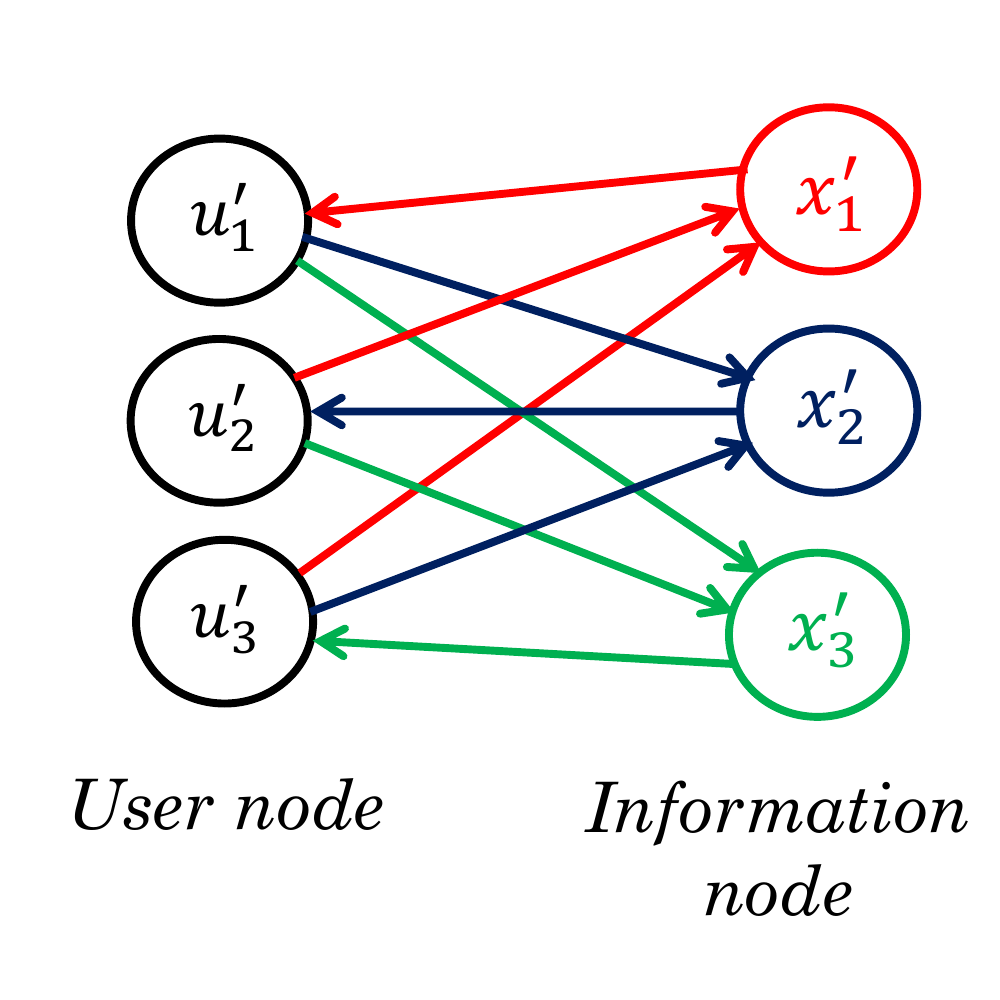}
\captionsetup{justification=centering}
\vspace{-5mm}
\caption{\footnotesize{Directed bipartite graph $\mathcal{B}_{IC}$ for the Index Coding problem in Example~\ref{exmp_IC}}.}
\label{fig:image_IC}
\vspace{-3mm}
\end{figure} 

\begin{example} \label{exmp_IC}
Consider the Index Coding problem with \mbox{$n_{IC}=3$} information symbols, \mbox{$m_{IC}=3$} users and user side information index sets \mbox{$\xcal_{1,IC}=\{2,3\}$}, \mbox{$\xcal_{2,IC}=\{1,3\}$}, \mbox{$\xcal_{3,IC}=\{1,2\}$}. The directed bipartite graph \mbox{$\mathcal{B}_{IC}=(\mathcal{U}_{IC},\mathcal{P}_{IC},\mathcal{E}_{IC})$} in Fig.~\ref{fig:image_IC} describes this scenario where \mbox{$\mathcal{U}_{IC}=\{u'_1,u'_2,u'_3\}$}, \mbox{$\mathcal{P}_{IC}=\{x'_1,x'_2,x'_3\}$} and $\mathcal{E}_{IC}=\{(x'_1,u'_1),(u'_1,x'_2),(u'_1,x'_3),(x'_2,u'_2),(u'_2,x'_1),(u'_2,x'_3),(x'_3,u'_3),(u'_3,x'_1),(u'_3,x'_2)\}$
\end{example}

\subsection{Construction of an Index Coding Problem from a given BNSI problem}    \label{ic_bnsi}

From the definition of $\mathcal{I}(q,m,n,\xcal,\delta_s)$ given in (\ref{ical_def}) and Theorem~\ref{thm1}, now we construct an $(\hat{m},n,\hat{\xcal},f)$ Index Coding problem from an $(m,n,\xcal,\delta_s)$ BNSI problem, where 
 \begin{equation} \label{hat_m}
 \hat{m}=
\left\{
\begin{array}{lc}
\sum_{i=1}^{m}\Comb{|\xcal_i|}{1} \times \Comb{|\xcal_i|-1}{2\delta_s-1} & \mbox{if}~|\xcal_i| \geq 2\delta_s \\ \sum_{i=1}^{m}\Comb{|\xcal_i|}{1} & \mbox{otherwise}
\end{array}
\right.
\end{equation}
and $\hat{\xcal}$ and $f$ are obtained from the construction of an Index Coding problem as described in Algorithm~3.
\begin{algorithm}[h]
\caption{Construction of an Index Coding problem from a given BNSI problem}
\SetAlgoLined
\textbf{Input}: $(m,n,\xcal,\delta_s)$ BNSI problem\\
\textbf{Output}: $(\hat{m},n,\hat{\xcal},f)$ Index Coding problem\\
 \% \% {Initialization:} $j=0$ \\
 \% \% Iteration: \\
 \For{$i=1,2,\dots,m$}{
 \For{each $p \in \xcal_i$}{
 \For{each $Q \subseteq \xcal_i \setminus \{p\}$ with $|Q|=\min\{|\xcal_i|-1,2\delta_s-1\}$}{
 $j \leftarrow j+1$\\
 $\xcal_{j,IC} \leftarrow \xcal_i \setminus (Q \cup \{p\})$\\
 $f(j) \leftarrow p$
 }
 }
 }
 $\hat{\xcal}=(\xcal_{1,IC},\xcal_{2,IC},\dots,\xcal_{\hat{m},IC})$\\
 \texttt{output} $(\hat{m},n,\hat{\xcal},f)$ Index Coding problem; \texttt{return};
\end{algorithm}

Algorithm~3 considers each $i^{th}$ user $u_i,~i \in [m]$ in $(m,n,\xcal,\delta_s)$ BNSI problem, for every possible choice of an element $p \in \xcal_i$ and a set $Q \subseteq \xcal_i \setminus \{p\}$ such that $|Q|=\min\{|\xcal_i|-1,2\delta_s-1\}$, it defines a new user $u'_j$ with $f(j)=p$ and $\xcal_{j,IC}=\xcal_i \setminus (Q \cup \{p\})$. In the newly constructed Index Coding problem, the total number of users $m_{IC}=\hat{m}$, number of information symbols $n_{IC}=n$, the tuple of side information index sets $\xcal_{IC}$ is given by $\hat{\xcal}$ and the demanded message $f(j)$ of each user $u_j,~j \in [\hat{m}]$ is given by mapping $f$. Hence we will obtain an $(\hat{m},n,\hat{\xcal},f)$ Index Coding problem. Now we relate the set $\mathcal{I}(q,m,n,\xcal,\delta_s)$ defined for the $(m,n,\xcal,\delta_s)$ BNSI problem and the set $\ical_{IC}(q,\hat{m},n,\hat{\xcal},f)$ for the $(\hat{m},n,\hat{\xcal},f)$ Index-Coding problem.

\begin{theorem}  \label{i_bnsi_i_ic}
Let for an $(m,n,\xcal,\delta_s)$ BNSI problem, the set $\mathcal{I}(q,m,n,\xcal,\delta_s)$ be defined by (\ref{ical_def}), and for the equivalent $(\hat{m},n,\hat{\xcal},f)$ Index Coding problem constructed from the $(m,n,\xcal,\delta_s)$ BNSI problem, the set $\ical_{IC}(q,\hat{m},n,\hat{\xcal},f)$ be defined by (\ref{i_ic}). Then $\mathcal{I}(q,m,n,\xcal,\delta_s)=\ical_{IC}(q,\hat{m},n,\hat{\xcal},f)$.   
\end{theorem}
\begin{proof}
To show that $\mathcal{I}(q,m,n,\xcal,\delta_s)=\ical_{IC}(q,\hat{m},n,\hat{\xcal},f)$, we will show that $\mathcal{I}(q,m,n,\xcal,\delta_s)\subseteq \ical_{IC}(q,\hat{m},n,\hat{\xcal},f)$ and $\ical_{IC}(q,\hat{m},n,\hat{\xcal},f) \subseteq \mathcal{I}(q,m,n,\xcal,\delta_s)$.

\vspace{2mm}
\textit{Proof for} $\mathcal{I}(q,m,n,\xcal,\delta_s) \subseteq \ical_{IC}(q,\hat{m},n,\hat{\xcal},f)$:
Suppose a vector $\zb \in \mathcal{I}(q,m,n,\xcal,\delta_s)$. Then from (\ref{ical_def}), we obtain that there exists at least one $i \in [m]$ such that $wt(\zb_{\xcal_i}) \in [2\delta_s]$. Therefore $\zb_{\xcal_i} \neq \mathbf{0}$. Hence there exists a $p \in \xcal_i$ such that $z_p \neq 0$. Note that $wt(\zb_{\xcal_i}) \leq 2\delta_s$ and since $wt(\zb_{\xcal_i \setminus \{p\}}) \leq 2\delta_s-1$ there exists $Q \subseteq \xcal_i \setminus \{p\}$ such that $|Q|=\min\{|\xcal_i|-1,2\delta_s-1\}$ and $\zb_{\xcal_i \setminus (Q \cup \{p\})}= \mathbf{0}$. Now using the construction procedure described in Algorithm~3 we see that there exists a $j$ such that $\xcal_{j,IC}=\xcal_i \setminus (Q \cup \{p\})$ satisfies $\zb_{\xcal_{j,IC}}=\mathbf{0}$ and and $f(j)=p$ satisfies $z_{f(j)} \neq 0$. Hence $\zb \in \ical_{IC}(q,\hat{m},n,\hat{\xcal},f)$.

\vspace{2mm}
\textit{Proof for} $\ical_{IC}(q,\hat{m},n,\hat{\xcal},f) \subseteq \mathcal{I}(q,m,n,\xcal,\delta_s)$: Suppose a vector $\zb \in \ical_{IC}(q,\hat{m},n,\hat{\xcal},f)$. Then there exists at least one user $j \in [\hat{m}]$ such that $z_{f(j)} \neq 0$ and $\zb_{\xcal_{j,IC}}=\mathbf{0}$.  Let $p=f(j)$ and $Q$ be any \mbox{$\min\{n-|\xcal_{j,IC}|-1,2\delta_s-1\}$} elements from the remaining set $[n] \setminus (\{f(j)\} \cup \xcal_{j,IC})$ . Note that $wt(\zb_{(\{p\}\cup \xcal_{j,IC} \cup Q))}) \in [2\delta_s]$. From Algorithm~3 we see that there exists $i \in [m]$ such that \mbox{$\xcal_i=\{p\} \cup \xcal_{j,IC} \cup Q$} satisfies $wt(\zb_{\xcal_i}) \in [2\delta_s]$. Hence $\zb \in \mathcal{I}(q,m,n,\xcal,\delta_s)$.

Hence the theorem holds. 
\end{proof}
 
\begin{example}  \label{BNSI_to_IC}
Here we will consider the BNSI problem scenario given in Example~\ref{exmp1}. The total number of users in corresponding Index-Coding problem will be $3 \times \Comb{3}{1} \times \Comb{2}{1}=18$ but among them only $12$ users have distinct \emph{(side information, demanded message)} pair. The number of information symbols will be same for Index Coding and BNSI problem. Table~\ref{table:1} shows all the distinct users of the Index Coding problem with their \emph{demanded message} and \emph{side information} symbols and Fig.~\ref{fig:image2} shows the corresponding bipartite graph for the Index Coding problem.
\begin{figure}[h!]
\centering
  \parbox[t]{12cm}{\null
  \centering
  \vspace{4mm}
  \captionof{table}[t]{Users in Index Coding problem}
  \label{table:1}
  \vspace{2mm}
  \begin{tabular}{|c|c|c|}
  \hline
 \multirow{2}{4em}{\textit{Rx. index}} & \textit{demanded} & \textit{side} \\
 & \textit{message} & \textit{information} \\
\hline
$1$ & $x_1$ & $\{x_2\}$\\
\hline
$2$ & $x_1$ & $\{x_3\}$\\
\hline
$3$ & $x_1$ & $\{x_4\}$\\
\hline
$4$ & $x_2$ & $\{x_1\}$\\
\hline
$5$ & $x_2$ & $\{x_3\}$\\
\hline
$6$ & $x_2$ & $\{x_4\}$\\
\hline
$7$ & $x_3$ & $\{x_1\}$\\
\hline
$8$ & $x_3$ & $\{x_2\}$\\
\hline
$9$ & $x_3$ & $\{x_4\}$\\
\hline
$10$ & $x_4$ & $\{x_1\}$\\
\hline
$11$ & $x_4$ & $\{x_2\}$\\
\hline
$12$ & $x_4$ & $\{x_3\}$\\
\hline
\end{tabular}
  }
  \parbox[t]{6cm}{\null
  \centering
  \vspace{-7mm}
  \includegraphics[width=1.4in]{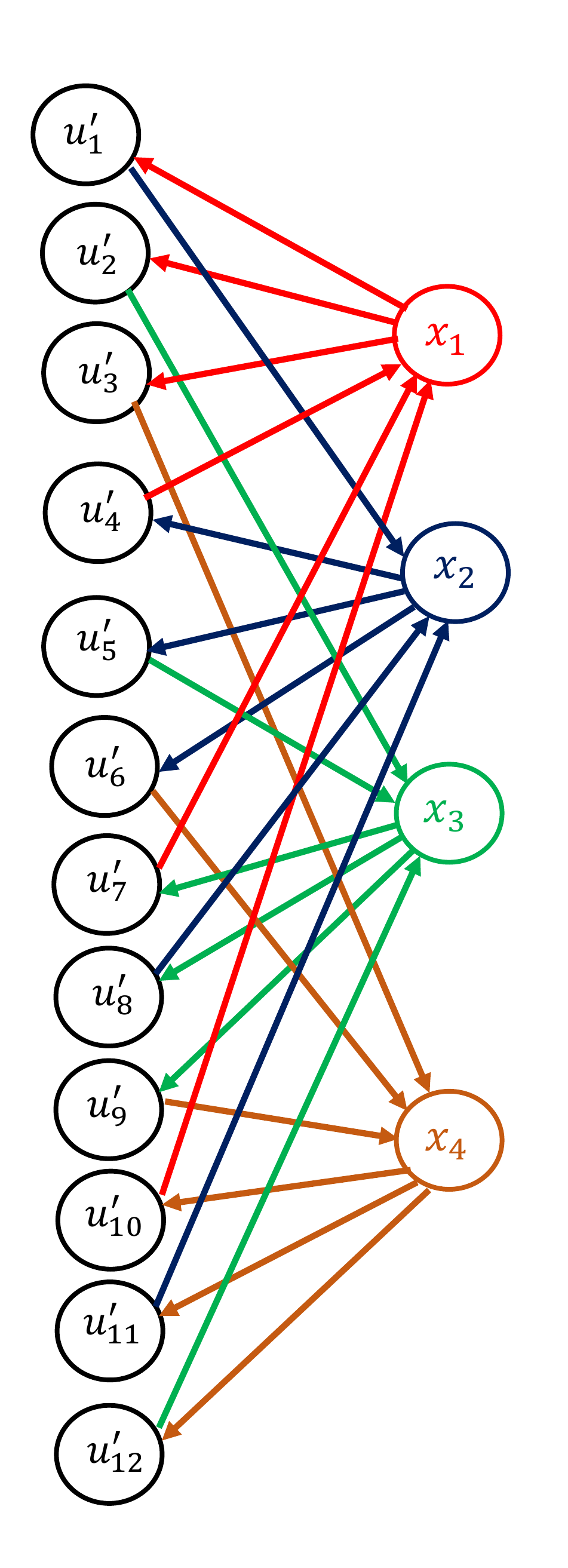}
  \captionof{figure}{Directed bipartite graph of Index Coding problem corresponding to BNSI problem in \mbox{Example \ref{exmp1}}.}
  \label{fig:image2}
}
\end{figure}
\end{example} 

Using the $(\hat{m},n,\hat{\xcal},\hat{f})$ Index-Coding problem corresponding to an $(m,n,\xcal,\delta_s)$ BNSI problem, now we relate the construction of $\Lb$ to the problem of designing scalar linear index coding scheme.
\begin{theorem}  \label{BNSI_IC_L}
$\Lb$ is a valid encoder matrix for the $(m,n,\xcal,\delta_s)$ BNSI problem if and only if $\Lb$ is a valid encoder matrix for the $(\hat{m},n,\hat{\xcal},\hat{f})$ scalar linear Index Coding problem.  
\end{theorem}
\begin{proof}
From Theorem~\ref{thm1} we know that $\Lb$ is a valid encoder matrix for the $(m,n,\xcal,\delta_s)$ BNSI problem if and only if it satisfies
\begin{equation*}
\zb\Lb \neq {\bf{0}},\quad \forall \zb \in \mathcal{I}(q,m,n,\xcal,\delta_s).
\end{equation*}
Now from Theorem \ref{i_bnsi_i_ic} we have $\mathcal{I}(q,m,n,\xcal,\delta_s) = \ical_{IC}(q,\hat{m},n,\hat{\xcal},f)$. So, $\Lb$ also satisfies, 
\begin{equation*}
\zb\Lb \neq {\bf{0}},\quad \forall \zb \in \ical_{IC}(q,\hat{m},n,\hat{\xcal},\hat{f}).
\end{equation*}
Therefore using (\ref{L_IC}) we can conclude that $\Lb$ is a valid encoder matrix for the $(\hat{m},n,\hat{\xcal},f)$ scalar linear Index Coding problem if and only if $\Lb$ is a valid encoder matrix for the $(m,n,\xcal,\delta_s)$ BNSI problem.
\end{proof}

Theorem \ref{BNSI_IC_L} claims that constructing an encoder matrix $\Lb$ for the $(m,n,\xcal,\delta_s)$ BNSI problem is equivalent to constructing an encoder matrix for the $(\hat{m},n,\hat{\xcal},f)$ scalar linear Index Coding problem. From \cite{YBJK_IEEE_IT_11, Dau_IEEE_J_IT_13}, we know that an encoder matrix for scalar linear Index Coding problem can be found by finding a matrix that fits its \emph{side information hypergraph} and the optimal length of a scalar linear Index Code equals the \emph{min-rank} of its side information hypergraph. 
\begin{example} 
Here we will consider the BNSI problem scenario given in Example~\ref{exmp1}. The users in the corresponding Index Coding problem is listed in Table \ref{table:1} and the graphical representation is given in Fig. \ref{fig:image2}. In the bipatite graph, we can notice that the edge sets $\{(x_1,u'_3),(u'_3,x_4),(x_4,u'_{10}),(u'_{10},x_1)\}$, $\{(x_2,u'_6),(u'_6,x_4),(x_4,u'_{11}),(u'_{11},x_2)\}$ and $\{(x_3,u'_9),(u'_9,x_4),(x_4,u'_{12}),(u'_{12},x_3)\}$ constitute $3$ cycles involving information symbol sets $\{x_1,x_4\}$, $\{x_2,x_4\}$ and $\{x_3,x_4\}$ respectively. Now using the \emph{Cyclic Code Actions} as described in \cite{MAZ_IEEE_IT_13} on each of these cycles, we can save one transmission. We encode the information symbols corresponding to $1^{st}$, $2^{nd}$ and $3^{rd}$ cycles as $x_1+x_4$, $x_2+x_4$, $x_3+x_4$ respectively. Therefore the codeword $(x_1+x_4,x_2+x_4,x_3+x_4)$ saves one transmission. Hence, $N_{q,opt,IC} \leq 3$.

Again we can notice that users $u'_3,u'_6,u'_9$ have $x_4$ as side information and each of the three users demands three distinct messages $x_1,x_2,x_3$, respectively. Therefore the encoder needs to encode $x_1,x_2,x_3$ such that with appropriate decoding functions $u'_3,u'_6,u'_9$ can decode $x_1,x_2,x_3$ respectively using the common side information $x_4$. Hence, $N_{q,opt,IC} \geq 3$. So, $N_{q,opt,IC}=3$. The encoder matrix that generates the codeword $(x_1+x_4,x_2+x_4,x_3+x_4)$ is   
  
\begin{equation*}
\Lb=
\begin{bmatrix}
1 & 0 & 0\\
0 & 1 & 0\\
0 & 0 & 1\\
1 & 1 & 1
\end{bmatrix}.
\end{equation*}
Note that this $\Lb$ we took in Example \ref{exmp2} to validate the design criterion of a valid encoder matrix for a $(m,n,\xcal,\delta_s)$ BNSI problem given in Example \ref{exmp1} and we have also used this $\Lb$ to describe the \emph{Syndrome Decoding} for $(m,n,\xcal,\delta_s)$ BNSI problem in Example \ref{exmp3}. Therefore the matrix $\Lb$ serves as a valid encoder matrix both for $(m,n,\xcal,\delta_s)$ BNSI problem given in Example \ref{exmp1} and corresponding $(\hat{m},n,\hat{\xcal},f)$ Index Coding problem given in Example \ref{BNSI_to_IC}.   
\end{example}


\subsection{Lower Bound on $N_{q,opt}$ based on Index Coding}   \label{lb_ic}
A lower bound on the optimal codelength $N_{q,opt}$ of an $(m,n,\xcal,\delta_s)$ BNSI problem can be derived based on its equivalent $(\hat{m},n,\hat{\xcal},f)$ scalar linear Index Coding problem. Suppose a directed bipartite graph $\mathcal{B}_{IC}=(\mathcal{U}_{IC},\mathcal{P}_{IC},\mathcal{E}_{IC})$ represents the $(\hat{m},n,\hat{\xcal},f)$ scalar linear Index Coding problem. Theorem~1 of \cite{MAZ_IEEE_IT_13} shows that if a directed bipartite graph $\mathcal{G}$ representing a scalar linear Index-Coding problem with $P$ information symbols is acyclic then its optimal codelength $N_{opt}(q,\mathcal{G})=P$. Now consider the bipartite graph $\mathcal{B}_{IC}$ and perform the \textit{pruning operations} given Section II-A in \cite{MAZ_IEEE_IT_13} to construct a subgraph $\mathcal{B}_{IC}^s=(\mathcal{U}_{IC}^s,\mathcal{P}_{IC}^s,\mathcal{E}_{IC}^s)$ which is an acyclic subgraph with information-set $\mathcal{P}_{IC}^s$. This leads to a lower bound on the optimal codelength of our BNSI problem.
\begin{lemma}
Let $\mathcal{B}_{IC}^s=(\mathcal{U}_{IC}^s,\mathcal{P}_{IC}^s,\mathcal{E}_{IC}^s)$ be any acyclic subgraph of $\bcal_{IC}$ induced by information-set $\mathcal{P}_{IC}^s$. Then the optimal codelength over $\Fb_q$ for the $(m,n,\xcal,\delta_s)$ BNSI problem  satisfies,
\begin{equation*}
N_{q,opt} \geq |\mathcal{P}_{IC}^s|.
\end{equation*}
\end{lemma} 
\begin{proof}
From equivalence relation of the BNSI problem and the scalar linear Index Coding problem described in Theorem \ref{BNSI_IC_L}, we have $N_{q,opt} = N_{q,opt,IC}(\mathcal{B}_{IC})$. As $\mathcal{B}_{IC}^s$ is a subgraph of $\bcal_{IC}$,  using Lemma~1 in \cite{MAZ_IEEE_IT_13} we have $N_{q,opt,IC}(\mathcal{B}_{IC}) \geq N_{q,opt,IC}(\mathcal{B}_{IC}^s)$.    
The directed bipartite graph $\mathcal{B}_{IC}^s$ is an acyclic subgraph with information-set $\mathcal{P}_{IC}^s$. 
Thereby using Theorem 1 of \cite{MAZ_IEEE_IT_13}, we have $N_{q,opt,IC}(\mathcal{B}_{IC}^s)=|\mathcal{P}_{IC}^s|$. Hence, $N_{q,opt}(\mathcal{B}) \geq |\mathcal{P}_{IC}^s|$    
\end{proof}

\section{Conclusions and discussions}
We derived a design criterion for linear coding schemes for BNSI problems, and identified the subset of problems where linear coding provides gains over uncoded transmission. Reduction in the codelength is achieved by jointly coding the information symbols to simultaneously meet the demands of all the receivers. We have derived lower bounds on the optimal codelength. We have shown a valid encoder matrix can be constructed from the transpose of a parity check matrix of linear error correcting codes. Based on the construction of a valid encoder matrix derived from MDS code, we found some upper bounds on optimal codelength. Codelength can be further reduced by partitioning any BNSI problems into many BNSI subproblems. We have shown that each BNSI problem is equivalent to an Index Coding problem. The presented results bring to light several questions regarding BNSI networks, such as evaluation of optimum code length $N_{q,opt}$, designing linear coding schemes that achieve this optimum length, designing schemes that admit low complexity decoding at the receivers, some efficient algorithms to find the presented lower and upper bounds and designing schemes for broadcasting in the presence of channel noise.


\bibliographystyle{IEEEtran}


\end{document}